\documentclass[12pt,draftcls,onecolumn]{IEEEtran}

\usepackage{cite}
\usepackage{graphicx}
\usepackage{amsmath,amssymb,amsfonts}
\usepackage{array}
\usepackage{flushend}
\usepackage{color}

\makeatletter%
\if@twocolumn%
\newcommand{\Figwidth}{\columnwidth}%
\def\twocolbreak{\nonumber\\ &}%
\else
\newcommand{\Figwidth}{4.5in}%
\def\twocolbreak{}%
\fi%
\makeatother%

\allowdisplaybreaks

\begin{document}
\title{Coherence Disparity in Broadcast and Multiple Access Channels}

\author{Mohamed Fadel, {\em Student Member, IEEE}, and Aria Nosratinia, {\em Fellow, IEEE}
\thanks{The authors are with the Department of Electrical Engineering, University of Texas at
    Dallas, Richardson, TX 75083-0688 USA, E-mail:
    mohamed.fadel@utdallas.edu;aria@utdallas.edu.}
\thanks{This work was supported in part by the grants CIF1219065 and CIF1527598 from the National Science Foundation.}
}
\maketitle

\newtheorem{theorem}{Theorem}
\newtheorem{lemma}{Lemma}
\newtheorem{remark}{Remark}

\def\A{A}
\def\B{\mathbf B}
\def\D{\mathbf D}
\def\Dset{\mathbb D}
\def\E{E}
\def\d{d}
\def\G{\mathbf G}
\def\g{\mathbf g}
\def\Gset{\mathcal G}
\def\H{\mathbf H}
\def\h{\mathbf h}
\def\Hset{\mathbb H}
\def\I{\mathbf I}
\def\J{J}
\def\Jset{\mathbb J}
\def\K{K}
\def\l{\ell}
\def\N{N}
\def\M{M}
\def\P{\mathbf P}
\def\Q{Q}
\def\Qset{\mathbb Q}
\def\R{R}
\def\q{q}
\def\r{r}
\def\S{\mathcal{S}}
\def\T{T}
\def\u{\mathbf u}
\def\U{\mathbf U}
\def\Uset{\mathbb U}
\def\v{\mathbf v}
\def\V{\mathbf V}
\def\W{\mathbf W}
\def\w{\mathbf w}
\def\X{\mathbf X}
\def\x{\mathbf x}
\def\Y{\mathbf Y}
\def\y{\mathbf y}
\def\Z{\mathbf Z}
\def\z{\mathbf z}
\def\bigO{O}
\def\littleO{o}
\def\ZeroMat{\mathbf{0}}


\begin{abstract}
Individual links in a wireless network may experience unequal fading
coherence times due to differences in mobility or scattering
environment, a practical scenario where the fundamental limits of
communication have been mostly unknown. This paper studies broadcast
and multiple access channels where multiple receivers experience
unequal fading block lengths, and channel state information (CSI) is
not available at the transmitter(s), or for free at any receiver. In
other words, the cost of acquiring CSI at the receiver is fully accounted for in the degrees of
freedom. In the broadcast channel, the method of product superposition
is employed to find the achievable degrees of freedom. We start with
unequal coherence intervals with integer ratios. As long as the
coherence time is at least twice the number of transmit and receive
antennas, these degrees of freedom meet the upper bound in four cases:
when the transmitter has fewer antennas than the receivers, when all
receivers have the same number of antennas, when the coherence time of
one receiver is much shorter than all others, or when all receivers
have identical block fading intervals. The degrees of freedom region
of the broadcast under identical coherence times was also previously
unknown and is settled by the results of this paper. The disparity of
coherence times leads to gains that are distinct from those arising
from other techniques such as spatial multiplexing or multi-user
diversity; this class of gains is denoted {\em coherence
  diversity}. The inner bounds are further extended to the case of
multiple receivers experiencing fading block lengths of arbitrary
ratio or alignment. Also, in the multiple access channel with unequal
coherence times, achievable and outer bounds on the degrees of freedom
are obtained.
\end{abstract}

\IEEEpeerreviewmaketitle

\section{Introduction}

In a wireless network, variations in node mobility and scattering environment may easily produce unequal link coherence times. But the performance limits of wireless networks under unequal link coherence times has been for the most part an open problem. 

Even under identical coherence times, understanding the performance limits of many wireless networks under block fading or related models has been far from trivial, with some key results under identical fading intervals being discovered only very recently.
For a two-receiver MISO broadcast channel with receive-side channel state information (CSIR) and finite precision transmit-side channel state information (CSIT), Lapidoth {\em et al.}~\cite{Lapidoth_capacity} conjectured that the degrees of freedom will collapse to unity under (non-singular) correlated fading. Tandon {\em et al.}~\cite{Tandon_synergistic} considered the broadcast channel with heterogeneous CSIT, i.e., the CSIT with respect to different links may be perfect, delayed, or non-existent. In this case,~\cite{Tandon_synergistic} conjectured a collapse of degrees of freedom for a two-receiver broadcast as long as CSIT with respect to one link is missing. The conjectures of Lapidoth {\em et al.} and Tandon {\em et al.} were settled in the positive by Davoodi and Jafar~\cite{Davoodi_aligned}, using the idea of aligned image sets~\cite{Korner_images}. Furthermore, Mohanty and Varanasi~\cite{Mohanty_degrees} developed an outer bound for a $\K$-receiver MISO broadcast channel  where there is CSIT with respect to some link gains and delayed CSIT with respect to other link gains. For the 3-receiver case, when there is perfect CSIT for one receiver and delayed CSIT for the other two,  a transmission scheme achieving $\frac{5}{3}$ sum degrees of freedom was found. For the same system, Amuru {\em et al.}~\cite{Amuru_degrees} proposed a transmission scheme that achieves $\frac{9}{5}$ sum degrees of freedom. A broadcast channel with delayed CSIT was studied by Maddah-Ali and Tse~\cite{Maddah_completely} and Vaze and Varanasi~\cite{Vaze_degrees_2}, demonstrating that even completely outdated channel feedback is still useful. A scenario of mixed CSIT (imperfect instantaneous and perfect delayed) was considered in~\cite{Chen_imperfect,Chen_degrees,Yang_degrees,Gou_optimal,Yi_degrees,Kerret_degrees,Chen_toward}. 

Huang {\em et al.}~\cite{Huang_degrees} studied a two-receiver broadcast channel with CSIR  but no CSIT under i.i.d. fast fading (all the receivers have coherence time of length 1), showing TDMA is degrees of freedom optimal. Vaze and Varanasi~\cite{Vaze_degree} extended this result to multiple receivers and to a wider class of fading distributions and fading dynamics (not including block fading). The results of~\cite{Huang_degrees,Vaze_degree} were based on the notion of {\em stochastic equivalence} of links with respect to the transmitter, an idea previously appearing in~\cite{Lapidoth_capacity,Jafar_isotropic}. 

For the broadcast channel, a summary of the results of this paper is as follows. We begin by settling the open problem of the degrees of freedom of the multi-receiver block-fading channel with identical fading intervals. We show that with CSIR but no CSIT, the degrees of freedom is limited to TDMA. In the absence of CSIR (and CSIT) we show that once again the degrees of freedom cannot be improved beyond TDMA. 

We then proceed to address {\em unequal} fading intervals, where the perspective for the availability and the cost of CSI is quite distinct from the case of equal fading intervals. Specifically, the normalized per-transmission cost of acquiring CSIR, e.g. via pilots~\cite{Hassibi_how}, is closely related to the block length, therefore the normalized cost of CSIR for links with unequal coherence times may vary widely. It follows that when fading intervals are unequal, any assumption of free CSIR may obscure important features of the problem. Therefore we adopt a model {\em without} free availability of CSIR, i.e., one where the cost of CSIR must be accounted for. For achievable degrees of freedom of the multi-receiver broadcast channel, we propose a generalization of the method of product superposition\footnote{Li and Nosratinia~\cite{Li_product, Li_coherent} introduced this method for the special case of a two-receiver broadcast channel where one receiver has a very long coherence time compared with the other.} to multiple receivers with coherence times of arbitrary integer ratios, and without free CSIR. Also, we do not assume availability of CSIT. This achievable rate is obtained by transmitting a pilot whenever one or more receivers experience a fading transition, and then during each pilot transmission exactly one (other) receiver who does not need the pilot can simultaneously utilize the channel for data transmission without contaminating the pilot. This leads to degrees of freedom gains that are directly tied to the disparity of coherence times, and are therefore called {\em coherence diversity}.

When the coherence time is at least twice the number of transmit and receive antennas, the obtained degrees of freedom are shown to meet the upper bound in four cases: When the number of transmit antennas is less than or equal to the number of antennas at every receiver, when all the receivers have the same number of antennas, when the coherence times of the receivers are very long compared to one receiver, or when all the receivers have identical coherence times. The development of outer bounds for this problem makes use of the idea of channel enhancement~\cite{Weingarten_capacity}, which in our case consists of increasing the coherence times of all receivers to match the coherence time of the slowest channel.

The inner bounds for coherence diversity are further extended to the case of multiple receivers experiencing fading block lengths of arbitrary ratio or alignment. Unaligned block fading intervals bring to mind the {\em blind interference alignment} of Jafar~\cite{Jafar_blind}. We consider a version of blind interference alignment that unlike~\cite{Jafar_blind} takes into account the full cost of CSIR via training; in that framework we explore the synergies between blind interference alignment and product superposition. 

For the block-fading {\em multiple access} channel, the capacity in the absence of CSIR is unknown.\footnote{In fact the capacity of point-to-point channel under this condition is also unknown except certain special cases~\cite{Marzetta_capacity}.} Shamai and Marzetta~\cite{Shamai_multiuser} conjectured that in the SIMO block-fading multiple access without CSIR, the sum capacity can be achieved by activating no more than $\T$ receivers. Also, for a two-receiver SISO multiple access channel with i.i.d. fast fading, a non-naive time-sharing inner bound and a cooperative outer bound on the capacity region were provided~\cite{Marina_rayleigh}. Furthermore, a multi-receiver multiple access channel with identical coherence times where the receivers are equipped with single antenna was considered in~\cite{Gopalan_calculating} where an inner bound on the network sum capacity was provided based on successive decoding, and an outer bound was obtained based on assuming cooperation between the transmitters. 

Our results for the multiple access channel are as follows: we begin by highlighting bounds on the degrees of freedom of the block-fading MIMO multiple access channel with {\em identical} coherence times in the absence of free CSIR, a result that is not complicated but has been absent from the literature. A conventional pilot-based scheme emitting individual and separate pilots from (a subset of) the transmitters antennas is considered that subsequently allows the receiver to perform zero-forcing. This method is shown to partially meet the cooperative outer bound. In particular, this method always achieves the optimal sum degrees of freedom, and in some cases is optimal throughout the degrees of freedom region. For the case of unequal coherence times, the same transmission technique is employed with pilots transmitted at the fading transition times of every active receiver. The outer bound is once again built on the concept of enhancing the channel ~\cite{Weingarten_capacity} by increasing the receivers coherence times so that the receivers of the enhanced channel have identical coherence times.

The key results of the paper are summarized in Table~\ref{table:DoF_bc} for broadcast channel and Table~\ref{table:DoF_mac} for multiple access channel, where $\N_i^{\ast} = \min\left\{\M, \N_i, \left\lfloor \frac{\T_i}{2}\right\rfloor\right\}$, $\N_{\max}=\max_{j \in \Jset} \left\{ \N_j \right\}$ and $j_{\min}$ is the receiver with the shortest coherence time in $\Jset$.

\renewcommand{\arraystretch}{2}
\begin{table*}
	\centering
\begin{tabular}{|l|l|}
 \hline
 Coherence Times     & \multicolumn{1}{c|}{Degrees of freedom} \\
 \hline
 Identical: $\T_k=\T, \ \forall k$     & $\sum_{i=1}^{\K} \frac{\d_i}{ \N_i^{\ast} \left( 1 - \frac{\N_i^{\ast}}{\T} \right) } \leq 1$\\
 \hline
 Integer ratio: $\frac{\T_k}{\T_{k-1}}\in {\mathbb Z}, \ \T_k \in {\mathbb N}$ 
 & Inner bound 1:
 $ \d_j = \begin{cases} \N_j^{\ast} \left( 1 - \frac{\N_j^{\ast}}{\T_j} - \frac{\min\left\{\M, \N_{\max}, \T_j \right\} - \N_j^{\ast}}{\T_{j+1}}\right), & \quad j=j_{\min} \\ \N_{j_{\min}}^{\ast} \min\{\M, \N_j, \T_{j_{\min}}\} \left( \frac{1}{\T_{j-1}}- \frac{1}{\T_j} \right),   & \quad j\neq j_{\min} \end{cases}$\\
 & Inner bound 2: 
$\d_j = \begin{cases} \N_j^{\ast} \left( 1 - \frac{\N_j^{\ast}}{\T_j} \right), & j=j_{\min} \\ \N_{j_{\min}}^{\ast} \min\{\N_{j_{\min}}^{\ast},\N_j^{\ast}\} \left(  \frac{1}{\T_{j-1}}- \frac{1}{\T_j} \right), & j\neq j_{\min} \end{cases}$ \\
& where $j \in \Jset \subseteq \left[ 1:\K \right]$ \\
& Outer bound:
$\sum_{j \in \Jset} \frac{\d_j}{\N_j^{\ast} \left( 1 - 
\frac{\N_j^{\ast}}{\T_{j_{\max}}}\right)} \leq 1, \quad \forall 
\Jset \subseteq \left[ 1:\K \right]$\\
& Cases of tightness: 
$ \begin{cases}
1) \M \leq \min_j \N_j \\
2) \N_j=\N, \quad \forall j  \\
3) \T_j \gg \T_1, \quad j=2, \cdots, \K \\
4) \T_j=\T, \quad \forall j
\end{cases}, \quad \T_j \geq 2 \max\{\M, \N_j \}
$ \\
 \hline
Arbitrary ratio: $\T_k \in {\mathbb N}, \ \forall k$ 
 & Inner bound:
 $ \d_j = \begin{cases} \N_j^{\ast} \left( 1 - \frac{\N_j^{\ast}}{\T_j} \right), & j=j_{\min} \\ \N_{j_{\min}}^{\ast} \min\{\N_{j_{\min}}^{\ast},\N_j^{\ast}\} \left(  \frac{1}{\T_{j-1}}- \frac{1}{\T_j} \right), & j\neq j_{\min} \end{cases}$\\
& where $j \in \Jset \subseteq \left[ 1:\K \right]$ \\
\hline
\end{tabular}
\caption{Degrees of freedom of block-fading broadcast channel with no CSI}
\label{table:DoF_bc}
\end{table*}

\renewcommand{\arraystretch}{2}
\begin{table*}
	\centering
\begin{tabular}{|l|l|}
 \hline
 Coherence Times      & \multicolumn{1}{c|}{Degrees of freedom} \\
 \hline
 Identical: $\T_k=\T, \ \forall k$    
 & Inner bound: $d_j = \M'_j \left( 1 - \frac{\sum_{j \in \Jset}{\M'_j}}{\T}\right), \quad j \in \Jset \subseteq [1:\K]$\\
 & Outer bound: $\sum_{j \in \Jset} \d_j \leq \min\left\{ \N, \sum_{j \in \Jset} \M_j \right\} \left( 1 - \frac{\min\left\{ \N, \sum_{j \in \Jset} \M_j \right\}}{\T}\right), \quad \forall \Jset \subseteq \left[ 1:\K \right]
$ \\
& Inner bound is tight against sum degrees of freedom \\ 
 \hline
 Integer ratio: $\frac{\T_k}{\T_{k-1}}\in {\mathbb Z}, \ \forall k$
 & Inner bound: $d_j = \M'_j \sum_{m=1}^\J{\left( \T_{i_1} - \sum_{n=1}^m{\M'_{i_n}}\right)\left( \frac{1}{\T_{i_m}} - \frac{1}{\T_{i_{m+1}}} \right)}$\\
& Outer bound: $\sum_{j \in \Jset} \d_j \leq \min\left\{ \N, \sum_{j \in \Jset} \M_j \right\} \left( 1 - 
\frac{\min\left\{ \N, \sum_{j \in \Jset} \M_j \right\}}{\T_{i_\J}}\right)$\\
& where $\Jset=\{i_1, \cdots, i_\J \} \subseteq \left[ 1:\K \right]$ \\
 \hline
\end{tabular}
\caption{Degrees of freedom of block-fading multiple access channel with no CSI}
\label{table:DoF_mac}
\end{table*}

\section{Broadcast Channel with Identical Coherence Times}
\label{section:identical_bc}

Consider a $\K$-receiver MIMO broadcast channel where the transmitter is equipped with $\M$ antennas and receiver~$k$ is equipped with $\N_k$ antennas, $k=1, \cdots, \K$. The signal at receiver~$k$ is
\begin{equation}\label{eq:received_bc}
\y_k(n) =\overline{\H}_k(n) \x(n) + \z_k(n), \qquad k=1,\cdots, \K,
\end{equation}
where $\x(n) \in \mathbb{C}^{\M \times 1}$ is the transmitted signal, $\z_k(n) \in \mathbb{C}^{\N_k \times 1}$ is receiver~$k$ i.i.d. Gaussian additive noise and $\overline{\H}_k(n) \in \mathbb{C}^{\N_k \times \M}$ is receiver~$k$ Rayleigh block-fading channel matrix with coherence time of length $\T_k$ time slots~\cite{Marzetta_capacity}, at the discrete time index $n$. One time slot is equivalent to a single transmission symbol period, and all $\T_k$ are positive integers. We assume no CSIT, meaning the realization of $\overline{\H}_k(n)$ is not known at the transmitter, whereas its distribution (including the length of the coherence time, and its transition) is globally known at the transmitter and at all receivers.

We assume that there are $\K$ independent messages associated with rates $\R_1(\rho), \cdots, \R_\K(\rho)$ to be communicated from the transmitter to the $\K$ receivers at $\rho$ signal-to-noise ratio. The degrees of freedom at receiver $k$ achieving rate $\R_k(\rho)$ can be defined as
\begin{equation}
\d_k=\lim_{\rho \rightarrow \infty} \frac{\R_k(\rho)}{\log (\rho)}.
\end{equation} 
The degrees of freedom region of a $\K$-receiver MIMO broadcast is defined as 
\begin{align}
\mathcal{D} =& \Bigg\{\big( \d_1, \cdots, \d_\K \big) \in \mathbb{R}_+^{\K} \bigg| \exists \big( \R_1(\rho), \cdots, \R_\K(\rho) \big) \in C(\rho), \twocolbreak
   \d_k = \lim_{\rho \rightarrow \infty} \frac{\R_k(\rho)}{\log (\rho)}, \quad k \in {1, \cdots, \K} \Bigg\},
\end{align} 
where $C(\rho)$ is the capacity region at $\rho$ signal-to-noise ratio.

Assume that the receivers have identical coherence times, where the coherence times are perfectly aligned, and furthermore, have the same length, namely $\T$. For the capacity to be determined, it is sufficient to study the capacity of only one coherence time. Define $\Y_k \in \mathbb{C}^{\N_k \times \T}$, $\X \in \mathbb{C}^{\M \times \T}$ to be the signal at receiver~$k=1, \ldots, \K$ and the transmitted signal, respectively, during the coherence time $\T$,
\begin{equation}
\Y_k = \H_k \X_k + \Z_k, \quad k=1, \cdots, \K,
\end{equation}
where $\H_k \in \mathbb{C}^{\N_k \times \M}$ is receiver~$k$ channel matrix which remains constant during the interval $\T$.

When there is CSIR, the degrees of freedom optimality of TDMA for two receivers with $\T=1$ was shown in~\cite{Huang_degrees}. Furthermore, the result was extended to arbitrary number of receivers and for a wider class of fading distribution~\cite{Vaze_degree}. Since there is no CSIT, and furthermore, the receivers have identical coherence times, namely $\T$, the receivers are stochastically equivalent (indistinguishable) with respect to the transmitter~\cite{Lapidoth_capacity,Jafar_isotropic}. As a result, TDMA is enough to achieve the degrees of freedom region of the system, i.e. the degrees of freedom region can be given by,
\begin{equation*}
\mathcal{D} = \Bigg\{ \left( \d_1, \cdots, \d_\K \right) \in \mathbb{R}_+^{\K} 
\bigg| \sum_{i=1}^{\K} \frac{\d_i}{ \min\{\M, \N_i\}} \leq 1 \Bigg\}.
\end{equation*}
In Appendix~\ref{appendix:identical_bc}, we extend this result for $\T \geq 1$ showing that TDMA is degrees of freedom optimal.

Now assume that, for a $\K$-receiver broadcast channel, there is no CSIR. As long as the receivers have identical coherence times, the receivers are still stochastically equivalent. In the sequel, we show that TDMA is enough to achieve the degrees of freedom region in this case.

\begin{theorem}
\label{theorem:identical_bc}
Consider a $\K$-receiver broadcast channel with identical coherence times $\T$. When there is no CSIT or CSIR meaning that the channel realization is not known, but the channel distribution is globally known, the degrees of freedom region of the channel is given by,
\begin{equation}
\mathcal{D} = \Bigg\{ \left( \d_1, \cdots, \d_\K \right) \in \mathbb{R}_+^{\K} 
\bigg| \sum_{i=1}^{\K} \frac{\d_i}{ \N_i^{\ast} \left( 1 - \frac{\N_i^{\ast}}{\T} 
\right) } \leq 1 \Bigg\}.
\label{eq:identical_bc}
\end{equation}
where $\N_i^{\ast} = \min \left\{\M, \N_i, \left\lfloor \frac{\T_i}{2} \right\rfloor \right\}$.
\end{theorem}

\begin{proof}
A simple time division multiplexing between the receivers achieves the degrees of freedom region. The remainder of the proof is dedicated to finding a corresponding outer bound. Without loss of generality, assume $\N_1 \leq \cdots \leq \N_\K$. When $\M \leq \N_1$, the cooperative outer bound~\cite{Sato_outer} for the sum degrees of freedom is
\begin{equation}
\sum_{i=1}^{\K} \d_i \leq \min \left\{\M, \left\lfloor \frac{\T}{2} \right\rfloor\right\} \left( 1 - \frac{\min\{\M, \left\lfloor \frac{\T}{2} \right\rfloor\}}{\T} \right), 
\end{equation}
which is tight against the TDMA inner bound. When $\M \geq \N_1$, to obtain the outer bound we need the following Lemma~\cite{Fadel_coherent}.

\begin{lemma}
\label{lemma:stoch_equ}
For the above $\K$-receiver broadcast channel, define $\overline{\Y} = \left[\Y_1^H, \; \Y_2^H, \; \cdots \;, \Y_{\K}^H \right]^H$ to be the matrix that contains all received signals during $\T$ interval, and $\overline{\Y}_{j} \in \mathbb{C}^{1 \times \T}$ is row $j$ of $\overline{\Y}$ and $\widetilde{\Y}_{\S}$ is the matrix constructed from excluding the set $\S$ of the rows from the matrix $\overline{\Y}$. Then we have
\begin{equation}
I \left( \X ; \overline{\Y}_j | U, \widetilde{\Y}_{\{j,\ell\}} \right) 
= I \left( \X ; \overline{\Y}_{\ell} | U, \widetilde{\Y}_{\{j,\ell\}} \right),
\label{eq:stoch_equ_1}
\end{equation}
and furthermore,
\begin{equation}
I \left( \X ; \overline{\Y}_j | U, \widetilde{\Y}_{\{j,\ell\}} \right) \geq I \left( \X ; \overline{\Y}_j | U, \widetilde{\Y}_{\{j,\ell\}}, \overline{Y}_{\ell} \right),
\label{eq:stoch_equ_2}
\end{equation}
where $U \rightarrow \X \rightarrow \overline{\Y}$ forms a Markov Chain.
\end{lemma}

Now, we are ready to find the outer bound for the case when $\M \geq \N_1$. Since the receivers have the same noise variance, the system is considered degraded~\cite{Bergmans_random, Bergmans_simple}, \cite[Section 5.7]{Gamal_network},
\begin{align}
\R_k \leq & I \left( U_k ; \Y_k | U^{k-1} \right), \quad k \neq \K, \nonumber \\
\R_\K \leq & I \left( \X  ; \Y_\K | U^{\K-1} \right),
\label{eq:degraded_bc_rates}
\end{align}
where $U_1 \rightarrow \cdots \rightarrow U_{\K-1} \rightarrow \X \rightarrow \left( \Y_1, \cdots \Y_\K \right)$ forms a Markov Chain, and $U_0$ is a trivial random variable. Using the chain rule, we can write~\eqref{eq:degraded_bc_rates} as
\begin{align*}
\R_k \leq & I \left( \X ; \Y_k | U^{k-1} \right) 
- I \left( \X ; \Y_k| U^k \right), \quad k \neq K, \nonumber\\
\R_\K \leq & I \left( \X ; \Y_\K | U^{\K-1} \right).
\end{align*}
Define $r_k$ to be the degrees of freedom of the term $I \left( \X ; \Y_k| U^k \right)$, where $ 0 \leq r_k \leq \N_k^{\ast} \left( 1 - \frac{\M}{\T}\right)$. Furthermore, the degrees of freedom of $I \left( \X ; \Y_1 \right)$ is bounded by the single receiver bound, i.e. $\N_1^{\ast} \left( 1 - \frac{\M}{\T}\right)$, hence,
\begin{align}
\R_1 \leq & \left( \N_1^{\ast} \left( 1 - \frac{\M}{\T}\right) - r_1 \right) 
\log \left( \rho \right) + \littleO(\log (\rho)), \nonumber \\ 
\R_k \leq & I \left( \X ; \Y_k | U^{k-1} \right) - 
r_k \log \left( \rho \right) + \littleO(\log (\rho)), \quad k \neq 1, K \nonumber\\
\R_\K \leq & I \left( \X ; \Y_\K | U^{\K-1} \right).
\end{align} 

Furthermore, we have 
\begin{align}
r_k \log(\rho) + \littleO(\log (\rho))=& I \left( \X ; \Y_k| U^k \right)      \nonumber \\
                                  \overset{(a)}{=} & I\left(\X ; \Y_{k,1:\N_k^{\ast}} | U^k \right)  \twocolbreak
                                                    +  I\left(\X ; \Y_{k,\N_k^{\ast}+1:\N_k} | U^k, \Y_{k,1:\N_k^{\ast}}\right)  
                                                   + \littleO(\log (\rho))   \nonumber \\
                               \overset{(b)}{\geq} & I\left(\X ; \Y_{k,1:\N_k^{\ast}} | U^k \right)  + \littleO(\log (\rho))   \nonumber \\
                                  \overset{(c)}{=} & \sum_{i=1}^{\N_k^{\ast}} I(\X ; \Y_{k,i} | U^k, \Y_{k,i+1:\N_k^{\ast}})  
                                  + \littleO(\log (\rho))    \nonumber \\
                                  \overset{(d)}{=} & \sum_{i=1}^{\N_k^{\ast}} I(\X ; \Y_{k,1} | U^k, \Y_{k,i+1:\N_k^{\ast}})   
                                  + \littleO(\log (\rho))                        \nonumber \\
                               \overset{(e)}{\geq} & \N_k^{\ast} I(\X ; \Y_{k,1} | U^k, \Y_{k,2:\N_k^{\ast}} ) 
                                  + \littleO(\log (\rho)),
\end{align}
where $\Y_{k,i:j}$ denotes the matrix constructed from the rows $i:j$ of the matrix $\Y_k$. $(a)$, and $(c)$ follow from the chain rule, and $(b)$ follows since mutual information is non-negative. Furthermore, $(d)$ follows from Lemma~\ref{lemma:stoch_equ} and $(e)$ follows since removing conditioning does not reduce the entropy. Therefore, 
\begin{equation}
I \left( \X ; \Y_{k,1} | U^k, \Y_{k,2:\N_k^{\ast}} \right) \leq 
\frac{r_k}{\N_k^{\ast}} \log \left( \rho \right) + \littleO(\log (\rho)). 
\label{eq:DoF_antenna_noncoh}
\end{equation} 
Furthermore,
\begin{align}
I \left( \X ; \Y_k | U^{k-1} \right)
\overset{(a)}{=} & I\left(\X ;\Y_{k,1:\N_k^{\ast}}|U^{k-1}\right) + I \left(\X ;\Y_{k,\N_k^{\ast}+1:\N_k}| U^{k-1}, \Y_{k,1:\N_k^{\ast}} \right) \nonumber\\
\overset{(b)}{=} & I\left(\X ;\Y_{k,1:\N_{k-1}^{\ast}}| U^{k-1} \right) \twocolbreak + I \left( \X ; \Y_{k,\N_{k-1}^{\ast}+1:\N_k^{\ast}}|U^{k-1}, \Y_{k,1:\N_{k-1}^{\ast}} \right) + \littleO(\log (\rho)) \nonumber \\
\overset{(c)}{=} & I\!\!\left(\X ;\Y_{k-1,1:\N_{k-1}^{\ast}}|U^{k-1} \right) \twocolbreak + \!I\!\!\left(\X ;\Y_{k,\N_{k-1}^{\ast}+1:\N_k^{\ast}}| U^{k-1},\Y_{k-1,1:\N_{k-1}^{\ast}} \right) \! + \!\littleO(\log (\rho)) \nonumber \\
= & r_{k-1} \log \left(\rho\right) \twocolbreak + \!\!\!\!\!\!\!\! \sum_{i=\N_{k-1}^{\ast}+1}^{\N_k^{\ast}} I\left(\X ;\Y_{k,i} |U^{k-1},\Y_{k-1,1:\N_{k-1}^{\ast}},\Y_{k,i+1:\N_k^{\ast}} \right) +\littleO(\log (\rho)) \nonumber\\
\overset{(d)}{\leq} & r_{k-1}\log \left(\rho \right) \twocolbreak + \left(\N_k^{\ast}-\N_{k-1}^{\ast}\right) I \left(\X ;\Y_{k-1,1}|U^{k-1},\Y_{k-1,2:\N_{k-1}^{\ast}}\right) + \littleO(\log (\rho)) \nonumber \\
\overset{(e)}{\leq} & r_{k-1} \log\left(\rho \right) +\left(\N_k^{\ast}-\N_{k-1}^{\ast}\right) \frac{r_{k-1}}{\N_{k-1}^{\ast}} \log \left( \rho \right) + \littleO(\log (\rho)) \nonumber \\
\leq & \frac{\N_k^{\ast}}{\N_{k-1}^{\ast}} r_{k-1} \log \left( \rho \right) + \littleO(\log (\rho)),
\end{align}
where $(a)$ and $(b)$ follow from applying the chain rule, and $ I \left(\X ; \Y_{k,\N_k^{\ast}+1:\N_k} | U^{k-1}, \Y_{k,1:\N_k^{\ast}} \right) = \littleO(\log (\rho))$ since more receive antennas than $\N_k^{\ast}$ does not increase the degrees of freedom \cite{Zheng_communication}. 
Furthermore, $(c)$ follows since $\Y_{k,1:\N_{k-1}^{\ast}}$ and 
$\Y_{k-1,1:\N_{k-1}^{\ast}}$ are statistically the same. $(d)$ 
follows from applying Lemma~\ref{lemma:stoch_equ} and $(e)$ follows 
from~\eqref{eq:DoF_antenna_noncoh}. Therefore,
\begin{align}
\d_1 \leq & \N_1^{\ast} \left( 1 - \frac{\M}{\T}\right) - r_1, \nonumber \\
\d_k \leq & \frac{\N_k^{\ast}}{\N_{k-1}^{\ast}} r_{k-1} - r_k, \quad i\neq1,\K, \nonumber \\
\d_\K \leq & \frac{\N_\K^{\ast}}{\N_{\K-1}^{\ast}} r_{\K-1}.
\end{align}
Hence,
\begin{align}
\sum_{i=1}^{\K} \frac{\d_i}{ \N_i^{\ast} \left( 1 - \frac{\M}{\T} \right) } &\leq 1 +
\sum_{i=2}^{\K} \frac{\r_{k-1}}{ \N_{k-1}^{\ast} \left( 1 - \frac{\M}{\T} \right) } -
\sum_{i=1}^{\K-1} \frac{\r_k}{ \N_k^{\ast} \left( 1 - \frac{\M}{\T} \right) }, \nonumber \\
& = 1,
\end{align}
where the last inequality follows since the two summations on the right hand side cancel each other. Thus, the degrees of freedom region is bounded by TDMA of the single receiver points $\N_k^{\ast} \left( 1 - \frac{\M}{\T} \right)$, which is maximized when $\M = \N_k^{\ast}$ \cite{Zheng_communication}, completing the proof of Theorem~\ref{theorem:identical_bc}.
\end{proof}

\section{Broadcast Channel with Heterogeneous Coherence Times}
\label{section:nonidentical_bc}

Consider the $\K$-receiver broadcast channel defined in~\eqref{eq:received_bc} where there is no CSIT or CSIR. The receivers have perfectly aligned coherence times with integer ratio, i.e., $\frac{\T_k}{\T_{k-1}}\in {\mathbb Z}, \forall k$. Fig.~\ref{fig:Coherence_Times_Three_Users} denotes three receivers where $\T_3= 2\T_2 = 4 \T_1$. In this system, the receivers are no longer stochastically equivalent, and hence, TDMA inner bound is no longer tight. 

The organization of this section is as follows. In Section~\ref{section:prod_super}, we revisit product superposition transmission introduced in~\cite{Li_coherent}. After that, in Section~\ref{section:nonidentical_ach_bc}, we give a product superposition transmission for the $\K$-receiver broadcast channel defined in~\eqref{eq:received_bc} calculating the achievable degrees of freedom region. Furthermore, we give an outer bound on the degrees of freedom region in Section~\ref{section:nonidentical_outer_bc}. We show the tightness of these bounds, and hence, the optimality of the achievable product superposition scheme for four cases in Section \ref{section:nonidentical_optimal_bc}. 
Finally, we give some numerical examples in Section~\ref{section:nonidentical_examples_bc}.

\begin{figure*}
\center
\includegraphics[width=4.in]{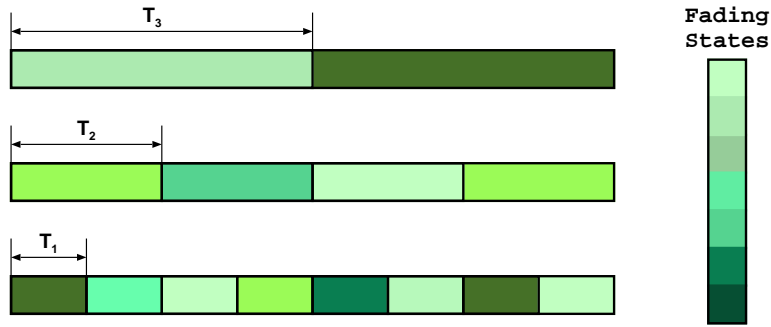}
\caption{Three receivers having aligned coherence times with integer ratio where $\T_3=2\T_2 = 4 \T_1$.}
\label{fig:Coherence_Times_Three_Users}
\end{figure*}
\subsection{Product Superposition Scheme}
\label{section:prod_super}
Li and Nosratinia~\cite{Li_product, Li_coherent} studied a two-receiver broadcast channel with no CSIT and with mixed CSIR; one static receiver has very long coherence time, hence, there is CSIR for this receiver, and one dynamic receiver has short coherence time $\T_d$, hence, there is no CSIR for this receiver. Li and Nosratinia showed that TDMA is \textit{suboptimal} in such a broadcast channel and proposed a product superposition scheme as follows. Consider $\M\geq\N_s\geq\N_d$, where $\N_s, \N_d$ are the numbers of antennas of the static and dynamic receivers, respectively. The transmitted signal is
\begin{equation}
\X=\X_s \X_d,
\end{equation}
where $\X_s \in \mathbb{C}^{\M \times \N_d}$ is the data matrix for the static receiver with i.i.d. $\mathcal{C}\mathcal{\N}\left(0,1\right)$ elements, and $\X_d \in \mathbb{C}^{\N_d \times \T_d}$ is the signal matrix for the dynamic receiver, where
\begin{equation}
\X_d = \left[ \I_{\N_d}, \; \X_{\delta}  \right],
\end{equation}
$\I_{\N_d}$ is $\N_d \times \N_d$ identity matrix, and $\X_{\delta} \in \mathbb{C}^{\N_d \times \left( \T_d - \N_d \right)}$ is the dynamic receiver data matrix having i.i.d $\mathcal{C}\mathcal{\N}\left(0,1\right)$ elements. Therefore the signal at the dynamic receiver, during $\T_d$ slots, is 
\begin{align}
\Y_d &= \H_d \X_s \left[ \I_{\N_d}, \; \X_{\delta}  \right] + \Z_d \nonumber \\ 
     &= \left[ \overline{\H}_d, \; \overline{\H}_d \X_{\delta}  \right] + \Z_d,
\end{align}
where $\overline{\H}_d=\H_d \X_s$, and $\H_d \in \mathbb{C}^{\N_d \times \M}$ is the dynamic receiver channel. The dynamic receiver estimates the equivalent channel $\overline{\H}_d$ during the first $\N_d$ slots and then decodes $\X_{\delta}$ {\em coherently}. On the other hand, the signal at the static receiver during the first $\N_d$ slots is
\begin{equation}
\Y_{s1} = \H_s \X_s + \Z_{s1},
\end{equation}
where $\H_s \in \mathbb{C}^{\N_s \times \M}$ is the static receiver channel which is known at the receiver, and hence, $\X_s$ can be decoded. As a result, the achievable degrees of freedom pair is 
\begin{equation}
\left( \N_d \left( 1 - \frac{\N_d}{\T_d} \right), \frac{\N_s\N_d}{\T_d} \right),
\end{equation}
which is strictly greater than TDMA. Thus, the product superposition achieves non-zero degrees of freedom for the static receiver ``for free'' in the sense that the dynamic receiver achieves the single-receiver degrees of freedom.

\subsection{Achievability}
\label{section:nonidentical_ach_bc}

\begin{theorem}
\label{theorem:nonidentical_ach_bc}
Consider a $\K$-receiver broadcast channel with heterogeneous coherence times and without CSIT or CSIR. The coherence times are perfectly aligned and integer multiples of each other, i.e., $\frac{\T_k}{\T_{k-1}}\in {\mathbb Z}$. Define $\Jset\subseteq \left[ 1:\K \right]$ to be a set of $\J$ receivers ordered ascendingly according to the coherence times length. For $j \in \Jset$, we can achieve the set of degrees of freedom tuples
$\Dset_1\left(\Jset\right):$
\begin{align} \label{eq:nonidentical_ach_bc_1}
\d_j &= 
\begin{cases}
\N_j^{\ast} \left( 1 - \frac{\N_j^{\ast}}{\T_j} - \frac{\min\left\{\M, \N_{\max}, \T_j\right\} - \N_j^{\ast}}{\T_{j+1}}\right), & j=j_{\min} \\ \N_{j_{\min}}^{\ast} \min\{\M, \N_j, \T_{j_{\min}}\} \left( \frac{1}{\T_{j-1}}- \frac{1}{\T_j} \right), & j\neq j_{\min}
\end{cases}.
\end{align}
Furthermore, we can achieve the set of degrees of freedom tuples $\Dset_2\left(\Jset\right):$
\begin{equation} \label{eq:nonidentical_ach_bc_2}
\d_j =  
\begin{cases}
\N_j^{\ast} \left( 1 - \frac{\N_j^{\ast}}{\T_j} \right),                               & j=j_{\min} \\
\N_{j_{\min}}^{\ast} \min\left\{ \N_j, \N_{j_{\min}}^{\ast} \right\} \left(  \frac{1}{\T_{j-1}}- \frac{1}{\T_j} \right), & j\neq j_{\min}
\end{cases},
\end{equation}
where $\N_j^{\ast} = \min\left\{\M, \N_j, \left\lfloor \frac{\T_j}{2} \right\rfloor \right\}$, $\N_{\max}=\max_{j \in \Jset} \left\{ \N_j \right\}$ and $j_{\min}$ is the receiver with the shortest coherence time in $\Jset$. The achievable degrees of freedom region is the convex hull of the degrees of freedom tuples, $\Dset_1\left(\Jset\right)$ and $\Dset_2\left(\Jset\right)$, over all the possible sets $\Jset \subseteq \left[ 1: \K\right]$, i.e.,
\begin{align}
\mathcal{D} = \left\{ \left( \d_1, \cdots, \d_\K \right) \in \text{Co} \left( \Dset_1\left(\Jset\right), \Dset_2\left(\Jset\right) \right),
 \forall {\Jset \subseteq \left[ 1: \K \right]} \right\}.
\label{eq:ach_bc_conv_hull}
\end{align}
\end{theorem}

\begin{proof}
The achievability proof is given in Section~\ref{section:nonidentical_ach_proof_bc}.
\end{proof}

\begin{remark}
$j_{\min}$ is the first receiver of $\Jset$ since the receivers of $\Jset$ are ordered ascendingly according to the coherence times length.
\end{remark}
\begin{remark}
The two achievable set of degrees of freedom tuples, $\Dset_1\left(\Jset\right)$ and $\Dset_2\left(\Jset\right)$, are achieved by product superposition transmission scheme. The degrees of freedom gains are different in the two sets due to the difference in the number of transmit antennas needed for channel estimation. Each set can construct a distinct achievable degrees of freedom region that does not include the other. In the first set, $\Dset_1\left(\Jset\right)$, all the receivers estimate the channel of the maximum number of antennas required for transmission, i.e., receiver~$j$ can estimate  the channel of $\N_j^{\ast}$ transmit antennas. In the second set, $\Dset_2\left(\Jset\right)$, the receivers are limited to estimate the channel of $\N_{j_{\min}}$ transmit antennas. For more details, the reader can be referred to the achievability proof given in Section~\ref{section:nonidentical_ach_proof_bc}.
\end{remark}
\begin{remark}
When the receivers have the same coherence times, product superposition transmission cannot achieve degrees of freedom gain. In this case, the degrees of freedom region is tight against TDMA.    
\end{remark}

\subsection{Outer Bound}
\label{section:nonidentical_outer_bc}

\begin{theorem}
\label{theorem:nonidentical_outer_bc}
Consider a $\K$-receiver broadcast channel under heterogeneous coherence times without CSIT or CSIR, meaning that the channel realization is not known, but the channel distribution is globally known. The coherence times are perfectly aligned and integer multiples of each others, i.e., $\frac{\T_k}{\T_{k-1}}\in {\mathbb Z}$. Define $\Jset \subseteq \left[ 1:\K \right]$ to be a set of $\J$ receivers ordered ascendingly according to the coherence times length, if a set of degrees of freedom tuples $\left( \d_1, \cdots, \d_K \right)$ is achievable, then it must satisfy the inequalities
\begin{equation}\label{eq:nonidentical_outer_bc}
\sum_{j \in \Jset} \frac{\d_j}{\N_j^{\ast} \left( 1 - 
\frac{\N_j^{\ast}}{\T_{j_{\max}}}\right)} \leq 1, \quad \forall 
\Jset \subseteq \left[ 1:\K \right],
\end{equation}
where $\N_j^{\ast}= \min \left\{\M, \N_j, \left\lfloor \frac{\T_j}{2} \right\rfloor \right\}$, and $j_{\max}$ is the receiver with the longest coherence time in $\Jset$.
\end{theorem}

\begin{remark}
The receivers of the set $\Jset$ are ordered ascendingly according to the coherence times length, i.e., $\frac{\T_k}{\T_{k-1}}\in {\mathbb Z}$. $\J_{\max}$ is the last receiver of the set, and $\T_{j_{\max}}$ is the longest coherence time in the set $\Jset$.
\end{remark}
\begin{proof}
We prove the Theorem by showing that for any $\Jset \subseteq \left[ 1: \K \right]$, the degrees of freedom are bounded by the inequality~\eqref{eq:nonidentical_outer_bc}. We show that for the set of receivers $\Jset$, increasing the coherence times of the receivers to be equal to the longest coherence time, i.e. $\T_j=\T_{j_{\max}}, \forall j \in \Jset$ cannot reduce the degrees of freedom. This means the degrees of freedom region of the resultant enhanced channel includes the original degrees of freedom region.

\begin{lemma}\label{lemma:enhance}
For a $\K$-receiver broadcast channel with heterogeneous coherence times and without CSIT or CSIR, define $\mathcal{D}\left(\Jset\right)$ to be the degrees of freedom region of a set of receivers $\Jset \subseteq \left[ 1:\K \right]$ where the receivers are ordered ascendingly according to the coherence times length. Define $\mathcal{\overline{D}}\left(\Jset\right)$ to be the degrees of freedom region of the same set of receivers $\Jset \subseteq \left[ 1:\K \right]$ where the receivers have the coherence time of the longest receiver, i.e., $\T_j= \T_{j_{\max}}, \forall j \in \Jset$. Thus, we have
\begin{equation}
\mathcal{D}\left(\Jset\right) \subseteq \mathcal{\overline{D}}\left(\Jset\right)
\end{equation}
\end{lemma}

\begin{proof}
See Appendix~\ref{appendix:enhance}.
\end{proof}

Using Lemma~\ref{lemma:enhance}, the degrees of freedom region for every set of receivers $\Jset \subseteq \left[ 1:\K \right]$ is included in the degrees of freedom region of an enhanced channel with identical coherence times of length $\T_{j_{\max}}$ slots. Furthermore, Theorem \ref{theorem:identical_bc} shows that the degrees of freedom region of the enhanced channel is tight against TDMA inner bound. Thus, we obtain the region in \eqref{eq:nonidentical_outer_bc}, and hence, the proof of Theorem \ref{theorem:nonidentical_outer_bc} is completed.
\end{proof}
\subsection{Optimality}
\label{section:nonidentical_optimal_bc}

For four cases, the achievable degrees of freedom region in Section~\ref{section:nonidentical_ach_bc} and the outer degrees of freedom region obtained in Section~\ref{section:nonidentical_outer_bc} are tight. In the four cases, the coherence time is at least twice the number of transmit and receive antennas, i.e., $\T_j \geq 2 \max \{\M, \N_j\}$.

\subsubsection{The transmitter has fewer antennas}

When $\M \leq \min_{j}\left\{ \N_j \right\}$, the outer degrees of freedom region given by~\eqref{eq:nonidentical_outer_bc} is
\begin{equation}
\sum_{j \in \Jset} \d_j \leq \M \left( 1 - \frac{ \M }{ \T_{j_{\max}} }\right), 
\quad \forall \Jset \subseteq \left[ 1:\K \right].
\end{equation}
The achievable degrees of freedom tuples in~\eqref{eq:nonidentical_ach_bc_2} are
\begin{equation}
\d_j = 
\begin{cases}
\M \left( 1 - \frac{\M}{\T_j} \right),                   & j=j_{\min} \\
\M^2 \left(  \frac{1}{\T_{j-1}}- \frac{1}{\T_j} \right), & j\neq j_{\min}
\end{cases}, \quad j \in \Jset.
\end{equation}
Hence,
\begin{align}
\sum_{j \in \Jset} \d_j &= \M \left( 1 - \frac{ \M }{\T_{j_{\min}}} \right) + \sum_{j \in \Jset, j\neq j_{\min}} \M^2 \left( \frac{ 1 }{\T_{j-1}} - \frac{ 1 }{\T_j} \right) \nonumber \\
&\overset{(a)}{=} \M \left( 1 - \frac{ \M }{\T_{j_{\min}}} \right) + \M^2 \left( \frac{1}{\T_{j_{\min}}} - \frac{ 1 }{\T_{j_{\max}}} \right) \twocolbreak
= \M \left( 1 - \frac{ \M }{\T_{j_{\max}}} \right),
\end{align}
where $(a)$ follows from the telescoping sum. Thus, the achievable degrees of freedom tuples are at the boundaries of the outer degrees of freedom region, consequently, the convex hull of the achievable degrees of freedom tuples is tight against the outer degrees of freedom region.

\subsubsection{The receivers have equal number of antennas}

When $ \N_k=\N, \forall k$, the outer degrees of freedom region given in~\eqref{eq:nonidentical_outer_bc} is
\begin{equation}
\sum_{j \in \Jset} \d_j \leq \N^{\ast} \left( 1 - \frac{ \N^{\ast} }{ \T_{j_{\max}} }\right), \quad \Jset \subseteq \left[ 1:\K \right].
\end{equation}
The achievable degrees of freedom tuples in~\eqref{eq:nonidentical_ach_bc_2} are 
\begin{equation}
\d_j =  
\begin{cases}
\N^{\ast} \left( 1 - \frac{\N^{\ast}}{\T_j} \right),            & j=j_{\min} \\
\N^{\ast^2} \left(  \frac{1}{\T_{j-1}}- \frac{1}{\T_j} \right), & j\neq j_{\min}
\end{cases}, \quad j \in \Jset.
\end{equation}
Hence,
\begin{align}
\sum_{j \in \Jset} \d_j &= \N^{\ast} \left( 1 - \frac{ \N^{\ast} }{\T_{j_{\min}}} \right) + \sum_{j \in \Jset, j \neq j_{\min}} \N^{\ast^2} \left( \frac{ 1 }{\T_{j-1}} - \frac{ 1 }{\T_j} \right) \nonumber \\
&\overset{(a)}{=} \N^{\ast} \left( 1 - \frac{ \N^{\ast} }{\T_{j_{\min}}} \right) + \N^{\ast^2} \left( \frac{1}{\T_{j_{\min}}} - \frac{ 1 }{\T_{j_{\max}}} \right) \twocolbreak
= \N^{\ast} \left( 1 - \frac{ \N^{\ast} }{\T_{j_{\max}}} \right).
\end{align}
The achievable degrees of freedom tuples are at the boundaries of the outer degrees of freedom region, thus, the outer degrees of freedom region is tight.

\subsubsection{The coherence times of the receivers are very large compared to the coherence time of one receiver}

When $\T_j \gg \T_1$, where $j=2, \cdots, \K$, the outer region given in~\eqref{eq:nonidentical_outer_bc} is
\begin{align}
                   \d_1 &\leq \N_1^{\ast} \left( 1 - \frac{ \N_1^{\ast} }{ \T_1 }\right), \nonumber \\
\sum_{j \in \Jset} \frac{\d_j}{\N_j^{\ast}} &\leq 1, \quad \Jset \subseteq \left[ 1:\K \right].
\end{align}
The achievable degrees of freedom tuples in~\eqref{eq:nonidentical_ach_bc_1}, $\Dset_1\left(\Jset\right)$, are
\begin{equation}
\d_j = 
\begin{cases}
\N_j^{\ast} \left( 1 - \frac{\N_j^{\ast}}{\T_j} \right),    & j=j_{\min} \\
\frac{\N^{\ast}_{j_{\min}} \N^{\ast}_j}{\T_{j-1}},      & j\neq j_{\min}
\end{cases}, \quad j \in \Jset.
\end{equation}
Therefore, 
\begin{align}
\sum_{j \in \Jset} \frac{\d_j}{\N_j^{\ast}} &\approx 1 - \frac{ \N_{j_{\min}}^{\ast} }{\T_{j_{\min}}} + \frac{ \N_{j_{\min}}^{\ast} }{\T_{j_{\min}}} \nonumber \\
&= 1, 
\end{align}
which means the achievable degrees of freedom region is tight.

\subsubsection{The receivers have identical coherence times}

In the case of identical coherence times, we showed in Section~\ref{section:identical_bc} that the degrees of freedom region is tight against TDMA. When $ \T_k = \T, \forall k$, the outer region given in~\eqref{eq:nonidentical_outer_bc} is
\begin{equation}
\sum_{j \in \Jset} \frac{\d_j}{\N_j^{\ast} \left( 1 - \frac{\N_j^{\ast}}{\T}\right)} \leq 1, \quad \forall \Jset \subseteq \left[ 1:\K \right],
\end{equation}
which is the same as the TDMA degrees of freedom region. In this case, the achievable degrees of freedom tuples in~\eqref{eq:nonidentical_ach_bc_2}, $\Dset_2\left(\Jset\right)$, are reduced to that obtained by TDMA.

\begin{figure}
\center
\includegraphics[width=\Figwidth]{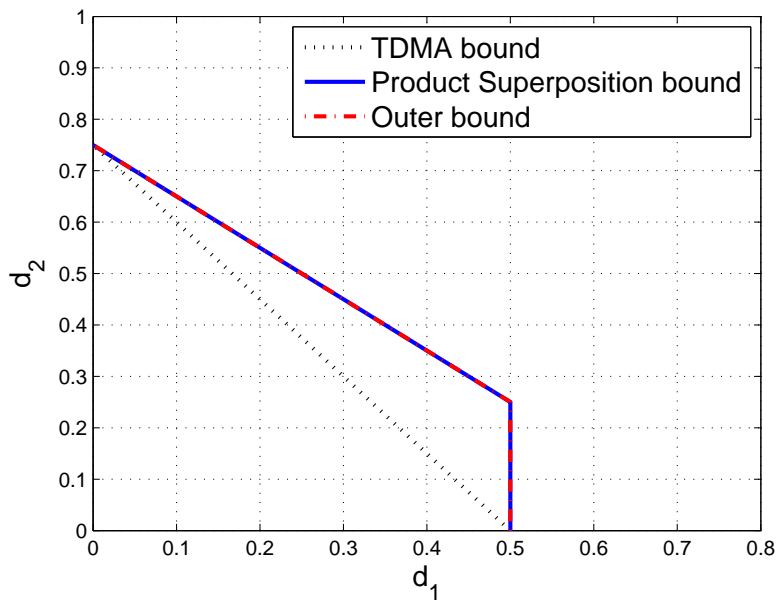}
\caption{Degrees of freedom region of a two-receiver broadcast channel with heterogeneous coherence times where $\M=\N_1=\N_2=1, \T_1=2, \T_2=4$.}
\label{fig:2users_siso_bc}
\end{figure}

\subsection{Numerical Examples}
\label{section:nonidentical_examples_bc}

Consider a single-antenna two-receiver broadcast channel, i.e. $\M=\N_1=\N_2=1$ with coherence times $\T_1=2$ and $\T_2=4$ slots. Thus, in this case, we have four possibilities of $\Jset: \left\{ \right\}, \left\{ 1 \right\}, \left\{ 2 \right\}, \left\{ 1, 2 \right\}$. According to Theorem~\ref{theorem:nonidentical_outer_bc}, the outer degrees of freedom region is given by 
\begin{align*}
\d_1        &\leq \frac{1}{2}, \\
\d_1 + \d_2 &\leq \frac{3}{4}.
\end{align*}
The achievable degrees of freedom tuples
\begin{equation}
\Dset_1\left(\Jset\right) = \Dset_2\left(\Jset\right): \left( 0, 0 \right), \left( \frac{1}{2}, 0 \right), \left( 0, \frac{3}{4}\right), \left( \frac{1}{2}, \frac{1}{4}\right).
\end{equation}
As shown in Fig.~\ref{fig:2users_siso_bc}, the outer and the achievable regions coincide on each other.

For a two-receiver broadcast channel with $\M=2, \N_1=1, \N_2=3$, and $\T_1=4, \T_2=24$, the outer degrees of freedom is given by 
\begin{align*}
      \d_1                    &\leq \frac{18}{24}, \\
\frac{\d_1}{23/24} + \frac{\d_2}{44/24} &\leq 1.
\end{align*}
Furthermore, 
\begin{align*}
\Dset_1\left(\Jset\right): & \left( 0, 0 \right), \left( \frac{18}{24}, 0 \right), \left( 0, \frac{44}{24}\right), \left( \frac{17}{24}, \frac{10}{24}\right), \\
\Dset_2\left(\Jset\right): & \left( 0, 0 \right), \left( \frac{18}{24}, 0 \right), \left( 0, \frac{44}{24}\right), \left( \frac{18}{24}, \frac{5}{24}\right).
\end{align*}
Fig.~\ref{fig:2users_mimo_bc} shows the gap between the outer and the achievable bounds.
\begin{figure}
\center
\includegraphics[width=\Figwidth]{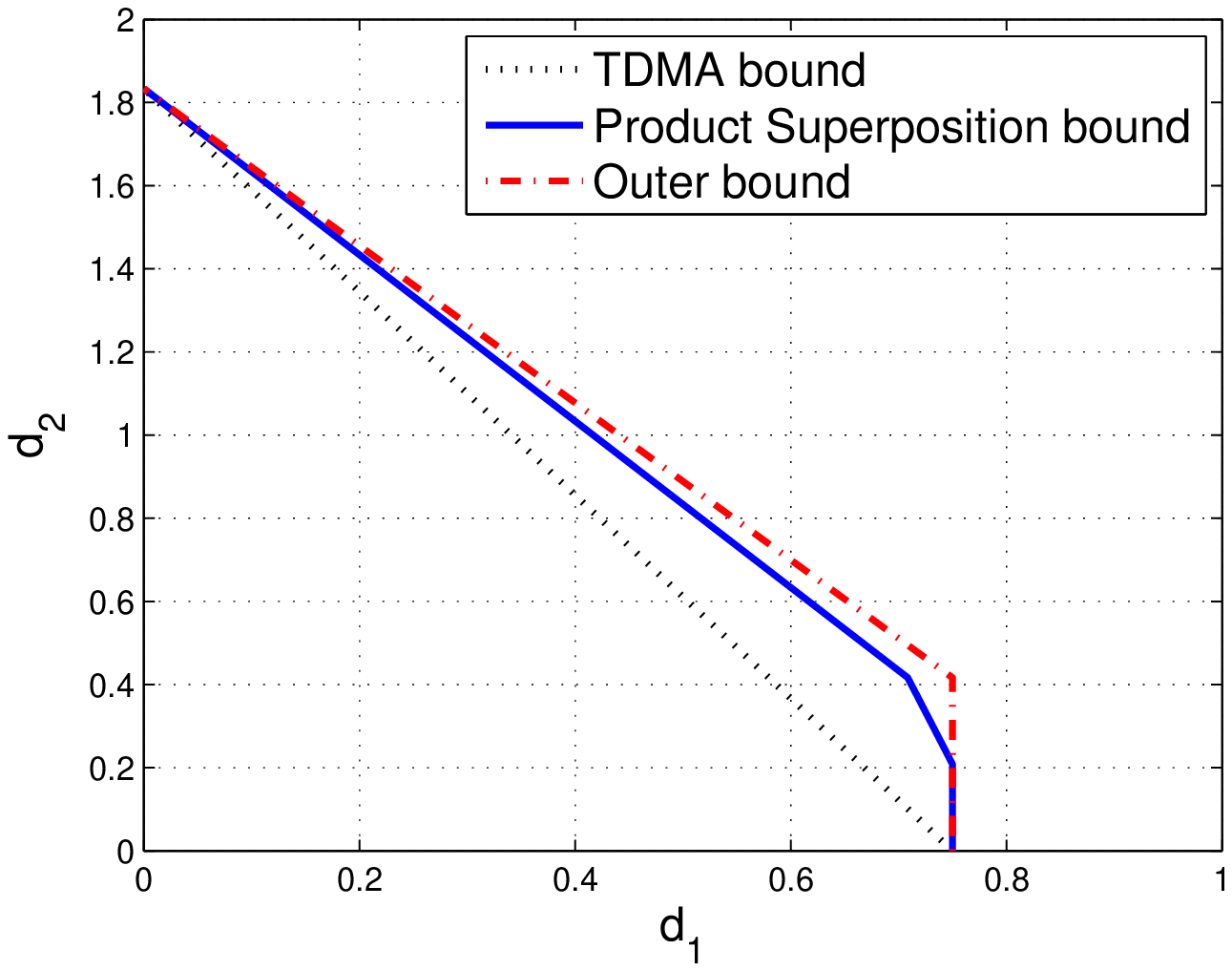}
\caption{Degrees of freedom region of a two-receiver broadcast channel with heterogeneous coherence times where $\M=2, \N_1=1, \N_2=3, \T_1=4, \T_2=24$.}
\label{fig:2users_mimo_bc}
\end{figure}

Furthermore, for a two-receiver broadcast channel with $\M=2, \N_1=1, \N_2=3$ and $\T_1=4$ and $\T_2=40$, the outer degrees of freedom region is given by 
\begin{align*}
      \d_1                    &\leq \frac{30}{40}, \\
\frac{\d_1}{39/40} + \frac{\d_2}{76/40} &\leq 1.
\end{align*}
For the achievable region in Theorem~\ref{theorem:nonidentical_ach_bc}, 
\begin{align*}
\Dset_1\left(\Jset\right): & \left( 0, 0 \right), \left( \frac{30}{40}, 0 \right), \left( 0, \frac{76}{40}\right), \left( \frac{30}{40}, \frac{9}{40}\right) \\
\Dset_2\left(\Jset\right): & \left( \frac{12}{16}, 0 \right), \left( 0, \frac{28}{16}\right), \left( \frac{29}{40}, \frac{18}{40}\right).
\end{align*}
Fig.~\ref{fig:2users_mimo_longT_bc} shows the gap between the achievable and the outer regions which is decreased with the increase of the ratio between the coherence times, $\frac{\T_2}{\T_1}$.
\begin{figure}
\center
\includegraphics[width=\Figwidth]{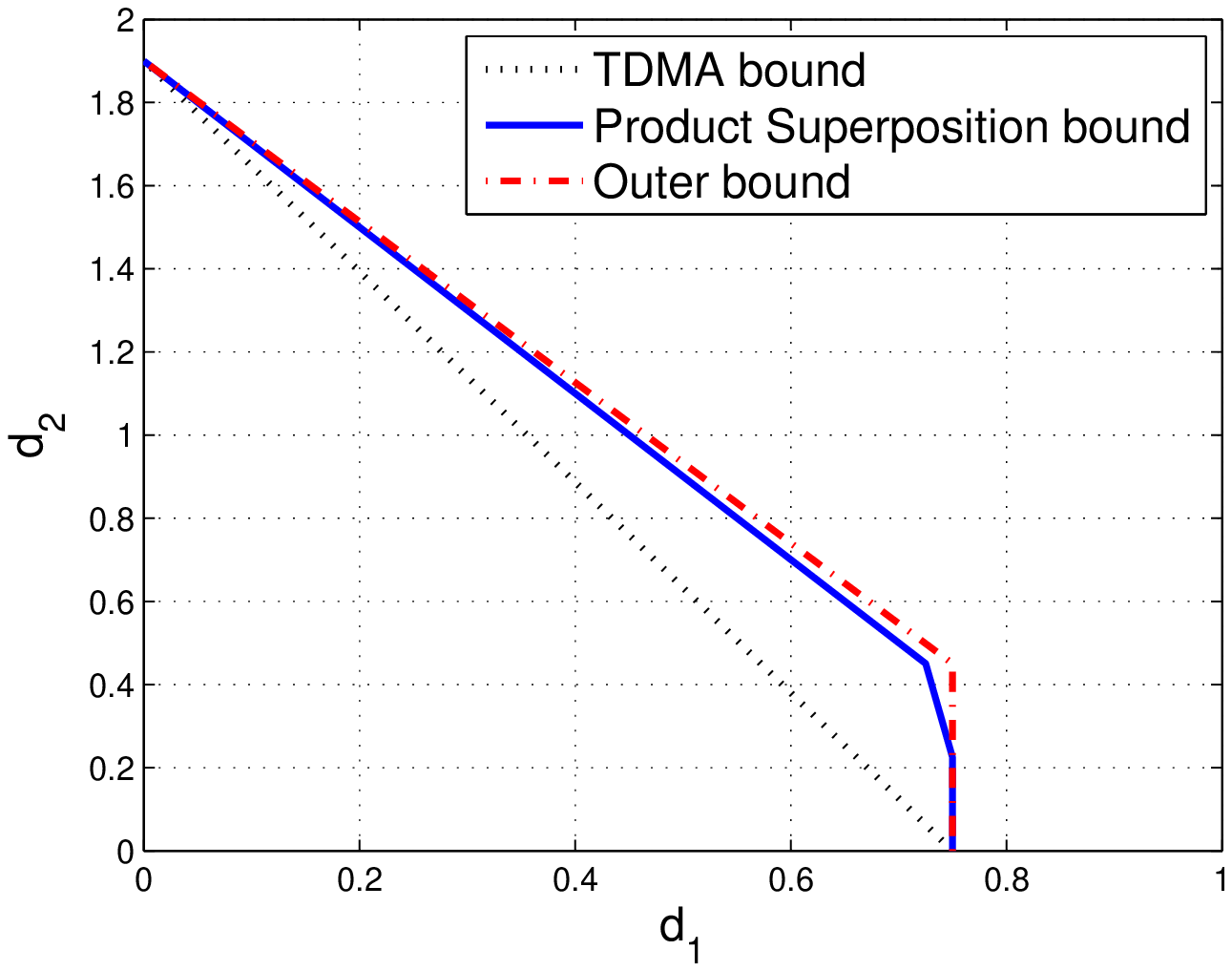}
\caption{Degrees of freedom region of a two-receiver broadcast channel with heterogeneous coherence times where $\M=2, \N_1=1, \N_2=3, \T_1=4, \T_2=40$.}
\label{fig:2users_mimo_longT_bc}
\end{figure}

\begin{figure}
\center
\includegraphics[width=\Figwidth]{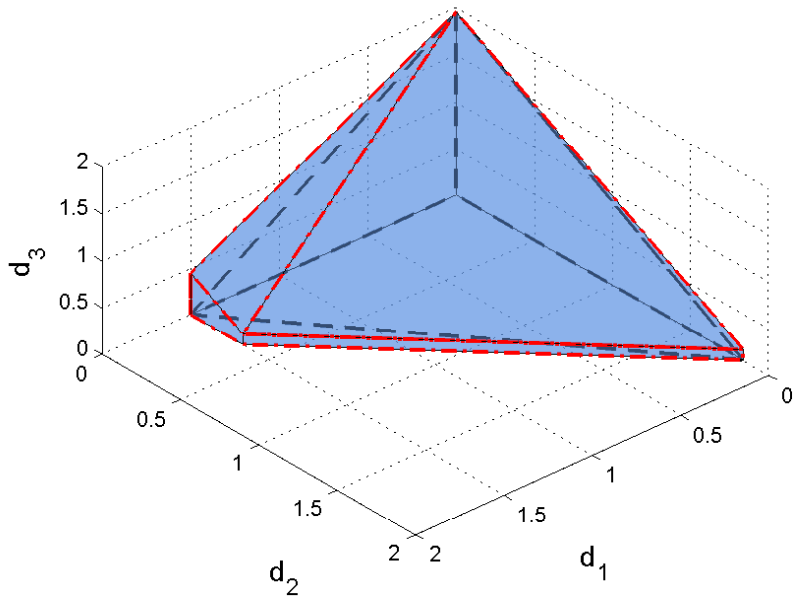}
\caption{Degrees of freedom region of a three-receiver broadcast channel with heterogeneous coherence times where $M=4, N_1=N_2=N_3=2, T_1=8, T_2=24, T_3=72$.}
\label{fig:3users_bc}
\end{figure}

Now consider a three-receiver broadcast channel with $\M=4, \N_1= \N_2=\N_3=2$ and $\T_1=6$, $\T_2=18, \T_3=54$. When the receivers have equal number of antennas, as discussed in Section~\ref{section:nonidentical_optimal_bc}, the achievable degrees of freedom and outer regions are tight. The outer degrees of freedom region is
\begin{align*}
      \d_1                    &\leq \frac{5}{6}, \\
\frac{\d_1}{17/18} + \frac{\d_2}{32/18} &\leq 1, \\
\frac{\d_1}{53/54} + \frac{\d_2}{104/54} + \frac{\d_3}{153/54} &\leq 1.
\end{align*}
For the achievable degrees of freedom region, we have 8 possibilities for $\Jset$ : 
\begin{equation*}
\left\{ \right\}, \left\{ 1 \right\}, \left\{ 2 \right\}, \left\{ 3 \right\}, \left\{ 1, 2 \right\}, \left\{ 1, 3 \right\}, \left\{ 2, 3 \right\}, \left\{ 1, 2, 3 \right\}.
\end{equation*}
Hence, 
\begin{align*}
\Dset_1\left(\Jset\right): & \left( 0, 0, 0 \right), \left( \frac{5}{6}, 0, 0 \right), \left( 0, \frac{32}{18}, 0 \right), \left( 0, 0, \frac{153}{154} \right), \left( \frac{14}{18}, \frac{4}{18}, 0 \right), \left( \frac{43}{54}, 0, \frac{24}{54} \right), \\ 
                           & \left( 0, \frac{94}{54}, \frac{12}{54} \right), \left( \frac{13}{18}, \frac{4}{18}, \frac{6}{54} \right). \\
\Dset_2\left(\Jset\right): & \left( 0, 0, 0 \right), \left( \frac{5}{6}, 0, 0 \right), \left( 0, \frac{32}{18}, 0 \right), \left( 0, 0, \frac{153}{154} \right), \left( \frac{5}{6}, \frac{2}{18}, 0 \right), \left( \frac{5}{6}, 0, \frac{8}{54} \right), \\ 
                           & \left( 0, \frac{32}{18}, \frac{8}{54} \right), \left( \frac{5}{6}, \frac{2}{18}, \frac{2}{54} \right).		
\end{align*}

Fig.~\ref{fig:3users_bc} shows the achievable degrees of freedom region (denoted by blue), the TDMA achievable region (denoted by black), and furthermore, the tight outer degrees of freedom region (denoted by red).

\section{Proof for Theorem~\ref{theorem:nonidentical_ach_bc}}\label{section:nonidentical_ach_proof_bc}

Achievable rates under coherence diversity for a general $\K$-receiver broadcast channel are attained by finding the best opportunities to re-use certain slots. Because the number of such opportunities blows up with $\K$, the central idea and intuition behind finding such opportunities are not easily visible in the general case of $\K$ receivers, where the achievable rates are eventually described via an inductive process. To highlight the ideas and the intuition in the achievable rate methodology, we develop these ideas in the special case of 3 receivers, which is the smallest number of receivers where the full richness of these interactions manifest themselves. We then proceed to describe the $\K$-receiver result in its full generality.

\subsection{Achievability for Three Receivers}
\label{section:ach_3users_bc}

In the case of three receivers we have 8 possible receivers sets $\Jset$: one empty set, $\left\{ \right\}$, achieving the trivial degrees of freedom tuple $\left( 0, 0, 0 \right)$, three single receiver sets, $\left\{1\right\}, \left\{ 2 \right\}, \left\{3\right\}$, three two-receiver sets, $\left\{1, 2\right\}, \left\{1, 3 \right\}, \left\{2, 3\right\}$, and one three-receiver set $\left\{1, 2, 3\right\}$. In the sequel, we first show the achievability of $\Dset_1\left(\Jset\right)$ and after that we give the achievability of $\Dset_2\left(\Jset\right)$.

\subsubsection{$\Dset_1\left(\Jset\right)$ achievability}
\label{section:ach_3users_bc_1}

For the three single-receiver sets, we can achieve the three degrees of freedom tuples
\begin{align}\label{eq:ach_3users_points_bc_1_1}
&\left( \N_1^{\ast} \left( 1 - \frac{\N_1^{\ast}}{\T_1} \right), 0, 0 \right),
\left( 0, \N_2^{\ast} \left( 1 - \frac{\N_2^{\ast}}{\T_2} \right), 0 \right), \twocolbreak
\left( 0, 0, \N_3^{\ast} \left( 1 - \frac{\N_3^{\ast}}{\T_3} \right) \right),
\end{align}
by serving only one receiver while the other receivers remain unserved. In particular, for receiver $k=1, 2, 3$, every $\T_k$ slots, a training sequence is sent during $\N_k^{\ast}$ slots and then data for receiver~$k$ is sent during the remaining $(\T_k - \N_k^{\ast})$ slots. $\N_k^{\ast} \left( 1 - \frac{\N_k^{\ast}}{\T_k} \right)$ degrees of freedom are achieved for receiver~$k$, whereas the other receivers achieve zero degrees of freedom.

For the three two-receiver sets, two receivers are being served while the third receiver remains unserved. Using product superposition for two receivers, the degrees of freedom tuples are
\begin{align}
\!\!\!\!\!\Bigg( \!\!\N_1^{\ast} \left( 1 - \frac{\N_1^{\ast}}{\T_1} \right) \! - \frac{\N_1^{\ast} \left( \min\left\{\M, \max \{ \N_1,\N_2 \}, \T_1  \right\} - \N_1^{\ast} \right)}{\T_2}, \twocolbreak
 \N_1^{\ast} \min\{ \M, \N_2, \T_1\} \left( \frac{1}{\T_1}  - \frac{1}{\T_2}\right), \! 0 \!\Bigg), \label{eq:ach_3users_points_bc_4_1} \\
\!\!\!\!\!\Bigg( \!\!\N_1^{\ast} \left( 1 - \frac{\N_1^{\ast}}{\T_1} \right) \!- \frac{\N_1^{\ast} \left( \min\left\{\M, \max \{ \N_1,\N_3 \}, \T_1 \right\} - \N_1^{\ast} \right)}{\T_3}, 
0, \twocolbreak
 \N_1^{\ast} \min\{\M, \N_3, \T_1\} \left( \frac{1}{\T_1} - \frac{1}{\T_3}\right) \!\!\Bigg), \label{eq:ach_3users_points_bc_5_1} \\
\!\!\!\!\!\Bigg( \!\!0, \N_2^{\ast} \left( 1 - \frac{\N_2^{\ast}}{\T_2} \right) \! - \frac{\N_2^{\ast} \left( \min\left\{\M, \max \{ \N_2,\N_3 \}, \T_2 \right\} - \N_2^{\ast} \right)}{\T_3}, \twocolbreak
\N_2^{\ast} \min\{\M, \N_3, \T_2\} \left( \frac{1}{\T_2} - \frac{1}{\T_3}\right)\!\!\!\Bigg). \label{eq:ach_3users_points_bc_6_1}
\end{align}
To achieve~\eqref{eq:ach_3users_points_bc_4_1}, product superposition transmission is sent over $\frac{\T_2}{\T_1}$ coherence intervals of receiver~1 (each of length $\T_1$ slots) as follows.
\begin{itemize}
	\item During the first coherence interval, training is sent during 
	      $\min\left\{\M, \max \{ \N_1,\N_2 \}, \T_1 \right\}$ slots for receiver~1 and receiver~2 channel estimation.
	      After that, data for receiver~1 is sent during the 
	      following $(\T_1 - \min\left\{\M, \max \{ \N_1,\N_2 \}, \T_1 \right\})$ slots. Receiver~1 
	      achieves 
				\begin{equation*}
   			\N_1^{\ast} \left( \T_1 - \min\left\{\M, \max \{ \N_1,\N_2 \}, \T_1\right\} \right)
				\end{equation*}
	      degrees of freedom.
	\item During the remaining coherence intervals, the 
	      transmitter sends, every $\T_1$ slots,
	      \begin{equation}\label{eq:ach_3users_bc_12_1}
           \X^{(12)}_i= \left[ \V_{i}, \; \V_{i}\U_{i}\right], \quad i=1,
                        \cdots, \frac{\T_2}{\T_1}-1,
        \end{equation}
        where $\U_{i} \in \mathbb{C}^{\N_1^{\ast} \times \left( \T_1 -
        \N_1^{\ast} \right)}, \V_{i} \in \mathbb{C}^{\M \times \N_1^{\ast}}$ 
				are data matrices for receiver~1, and receiver~2, respectively. 
				Thus, receiver~1 estimates its equivalent channel 
        $\overline{\H}_{1,i}=\H_{1,i}\V_{i}$, and decodes 
        $\U_i$ achieving $\left( \frac{\T_2}{\T_1} - 1 \right)\N_1^{\ast}\left( 
        \T_1 - \N_1^{\ast}\right)$ degrees of freedom. Furthermore, the channel of 
        receiver~2 remains constant and known, hence, $\V_{i}$ can be 
        decoded coherently at receiver~2 achieving $\left( \frac{\T_2}{\T_1} - 1\right)\N_1^{\ast} \min\{\M, \N_2, \T_1\}$ degrees of freedom. When $\N_2 \geq \T_1$, receiver~2 estimates only $\T_1$ antennas during the first subinterval. 
\end{itemize}
Thus, by the above product superposition scheme, for every $\T_2$ slots, receiver~1 achieves 
\[
\frac{\T_2}{\T_1} \N_1^{\ast}\left( \T_1 - \N_1^{\ast}\right) - \N_1^{\ast}\left( \min\left\{\M, \max \{ \N_1,\N_2 \}, \T_1 \right\} - \N_1^{\ast}\right) 
\]
degrees of freedom, and furthermore, receiver~2 achieves $\left( \frac{\T_2}{\T_1} - 1\right)\N_1^{\ast} \min\{\M, \N_2, \T_1\}$ degrees of freedom obtaining~\eqref{eq:ach_3users_points_bc_4_1}.

For achieving~\eqref{eq:ach_3users_points_bc_5_1}, a product superposition transmission similar to above is used after exchanging receiver~2 with receiver~3, i.e. using $\T_3, \N_3$ instead of $\T_2, \N_2$, respectively. Thus, for every $\T_3$ slots, receiver~1 achieves
\[
\frac{\T_3}{\T_1} \N_1^{\ast}\left( \T_1 - \N_1^{\ast}\right) - \N_1^{\ast}\left( \min\left\{\M, \max \{ \N_1,\N_3 \}, \T_1 \right\} - \N_1^{\ast}\right)
\]
 degrees of freedom, and furthermore, receiver~3 achieves $\left( \frac{\T_3}{\T_1} - 1\right) \N_1^{\ast} \min\{\M, \N_3, \T_1\}$ degrees of freedom.

Furthermore, we can achieve~\eqref{eq:ach_3users_points_bc_6_1} by the same transmission strategy, yet, with respect to $\T_2$ and $\T_3$. Thus, every $\T_3$ slots, receiver~2 achieves
\[
\frac{\T_3}{\T_2} \N_2^{\ast}\left( \T_2 - \N_2^{\ast}\right) - \N_2^{\ast}\left( \min\left\{\M, \max \{ \N_2,\N_3 \}, \T_2\right\} - \N_2^{\ast}\right)
\]
 degrees of freedom, and furthermore, receiver~3 achieves $\left( \frac{\T_3}{\T_2} - 1\right) \N_2^{\ast} \min\{\M, \N_3, \T_2\}$ degrees of freedom.

Now the remaining degrees of freedom tuple is the one with the three-receiver set $\left\{1, 2, 3\right\}$. In this case, the achievable degrees of freedom tuples are
\begin{align*}
\Bigg( \N_1^{\ast} &\left( 1 - \frac{\N_1^{\ast}}{\T_1} \right) - 
\frac{\N_1^{\ast}\left( \min\left\{\M, \max \{ \N_1,\N_2,\N_3 \}, \T_1\right\} - \N_1^{\ast}\right)}{\T_2}, \\
&\N_1^{\ast} \min\{\M, \N_2, \T_1\} \left( \frac{1}{\T_1} - \frac{1}{\T_2}\right), 
\N_1^{\ast} \min\{\M, \N_3, \T_1\} \left( \frac{1}{\T_2} - \frac{1}{\T_3}\right) \Bigg),
\end{align*}
which can be achieved by product superposition over $\frac{\T_3}{\T_2}$ coherence intervals of receiver~2 (each of length $\T_2$ slots) as follows.
\begin{itemize}
	\item During the coherence interval, the transmitted signal is the same as that 
	      used to achieve~\eqref{eq:ach_3users_points_bc_4_1}. Thus, receiver~1 achieves $\frac{\T_2}{\T_1} \N_1^{\ast}\left( \T_1 - \N_1^{\ast}\right) - \N_1^{\ast}\left( \min\left\{\M, \max \{ \N_1,\N_2 \}, \T_1 \right\} - \N_1^{\ast}\right)$ degrees of freedom, receiver~2 achieves $\left( \frac{\T_2}{\T_1} - 1\right)\N_1^{\ast}\min\{\N_2^{\ast}, \T_1\}$ degrees of freedom, and furthermore, receiver~3 estimates its channel.
	\item During the remaining $\left( \frac{\T_3}{\T_2} - 1\right)$ intervals, the transmitter sends, every $\T_2$-length subinterval, the same signal that achieves~\eqref{eq:ach_3users_points_bc_4_1} after multiplying it from the left by $\W_i$ which contains data for receiver~3. Therefore, during the first $\T_1$ of every $\T_2$-length coherence interval, the transmitted signal is
	      \begin{equation}
          \X^{(123)}= \left[ \W_i, \; \W_i\U_i\right].
        \end{equation}
	      After that during $\left( \frac{\T_2}{\T_1} - 1\right)\T_1$ slots, the transmitted signal is 
	      \begin{equation}
          \widetilde{\X}^{(123)}= \left[ \W_i\V_i, \; \W_i\V_i\U_i\right].
        \end{equation}
        receiver~1 estimates the equivalent channel $\overline{\overline{\H}}_{1,i}=\H_{1,i}\W_i\V_i$, and decodes $\U_i$, receiver~2 estimates $\H_{2,i}=\H_{2,i}\W_i$, and decodes $\V_i$ and receiver~3 decodes $\W_i$. Thus, the receivers achieve, respectively, $\left( \frac{\T_3}{\T_2} - 1 \right)\frac{\T_2}{\T_1}\N_1^{\ast}\left( \T_1 - \N_1^{\ast}\right)$, $\left( \frac{\T_3}{\T_2} - 1\right)\left( \frac{\T_2}{\T_1} - 1\right)\N_1^{\ast} \min\{\M, \N_2, \T_1\}$, and $\left( \frac{\T_3}{\T_2} - 1\right)\N_1^{\ast} \min\{\M, \N_3, \T_1\}$ degrees of freedom.
\end{itemize}

\subsubsection{$\Dset_2\left(\Jset\right)$ achievability}
\label{section:ach_3users_bc_2}

Similar to $\Dset_1\left(\Jset\right)$, we can achieve the degrees of freedom tuples~\eqref{eq:ach_3users_points_bc_1_1} that correspond to the three single-receiver sets by serving only one receiver while the other receivers remain unserved.

The degrees of freedom tuples of the three two-receiver sets are
\begin{align}
&\left( \N_1^{\ast} \left( 1 - \frac{\N_1^{\ast}}{\T_1} \right), 
\N_1^{\ast} \min\left\{ \N_1^{\ast}, \N_2 \right\} \left( \frac{1}{\T_1} - \frac{1}{\T_2}\right), 0 \right), \label{eq:ach_3users_points_bc_4_2}\\
&\left( \N_1^{\ast} \left( 1 - \frac{\N_1^{\ast}}{\T_1} \right), 
0, \N_1^{\ast} \min\left\{ \N_1^{\ast}, \N_3 \right\} \left( \frac{1}{\T_1} - \frac{1}{\T_3}\right) \right), \label{eq:ach_3users_points_bc_5_2} \\
&\left( 0, \N_2^{\ast} \left( 1 - \frac{\N_2^{\ast}}{\T_2} \right), 
\N_2^{\ast} \min\left\{ \N_2^{\ast}, \N_3 \right\} \left( \frac{1}{\T_2} - \frac{1}{\T_3}\right)\right). \label{eq:ach_3users_points_bc_6_2}
\end{align}
To achieve~\eqref{eq:ach_3users_points_bc_4_2}, product superposition is sent over $\frac{\T_2}{\T_1}$ coherence intervals of receiver~1 as follows.
\begin{itemize}
	\item During the first coherence interval, training is sent during $\N_1^{\ast}$ slots and data for receiver~1 is sent during the following $(\T_1 - \N_1^{\ast})$ slots. Receiver~1 achieves $\N_1^{\ast} \left( \T_1 - \N_1^{\ast} \right)$ degrees of freedom, and receiver~2 estimates its channel between $\min\left\{ \N_1^{\ast}, \N_2 \right\}$ transmit antennas.
	\item During the remaining coherence intervals, every $\T_1$ slots, the transmitter sends	      
	      \begin{equation}\label{eq:ach_3users_bc_12_2}
           \X^{(12)}_i= \left[ \V_{i}, \; \V_{i}\U_{i}\right], \quad i=1,
                        \cdots, \frac{\T_2}{\T_1}-1,
        \end{equation}
Thus, receivers achieve $\left( \frac{\T_2}{\T_1} - 1 \right)\N_1^{\ast} \left(\T_1 - \N_1^{\ast}\right)$, and $\left( \frac{\T_2}{\T_1} - 1\right)\N_1^{\ast}\min\left\{ \N_1^{\ast}, \N_2 \right\}$ degrees of freedom, respectively.
\end{itemize}
Thus, by the above product superposition transmission, for every $\T_2$ slots, receiver~1, and receiver~2 achieve $\frac{\T_2}{\T_1} \N_1^{\ast}\left( \T_1 - \N_1^{\ast}\right)$, and $\left( \frac{\T_2}{\T_1} - 1\right)\N_1^{\ast}\min\left\{ \N_1^{\ast}, \N_2 \right\}$ degrees of freedom, respectively, achieving~\eqref{eq:ach_3users_points_bc_4_2}.

For achieving~\eqref{eq:ach_3users_points_bc_5_2}, we use the same transmission scheme of achieving~\eqref{eq:ach_3users_points_bc_4_2} with respect to receiver~1 and receiver~3, i.e. replacing $\T_2, \min\left\{ \N_1^{\ast}, \N_2 \right\}$ with $\T_3, \min\left\{ \N_1^{\ast}, \N_3 \right\}$, respectively. Thus, receiver~1, and receiver~3 achieve $\frac{\T_3}{\T_1} \N_1^{\ast}\left( \T_1 - \N_1^{\ast}\right)$, and $\left( \frac{\T_3}{\T_1} - 1\right)\N_1^{\ast}\min\left\{ \N_1^{\ast}, \N_3 \right\}$ degrees of freedom, respectively, for every $\T_3$ slots. Similarly, we can achieve~\eqref{eq:ach_3users_points_bc_6_2} by the same transmission strategy, yet, with respect to $\T_2$ and $\T_3$.

For the three-receiver set, the achievable degrees of freedom tuples are
\begin{align*}
\Bigg( \N_1^{\ast} & \left( 1 - \frac{\N_1^{\ast}}{\T_1} \right),
\N_1^{\ast} \min\left\{ \N_1^{\ast}, \N_2 \right\} \left( \frac{1}{\T_1} - \frac{1}{\T_2}\right), \twocolbreak
\N_1^{\ast} \min\left\{ \N_1^{\ast}, \N_3 \right\} \left( \frac{1}{\T_2} - \frac{1}{\T_3}\right) \Bigg), 
\end{align*}
which can be achieved by product superposition transmission for the three receivers over $\frac{\T_3}{\T_2}$ coherence intervals of receiver~2 as follows.
\begin{itemize}
	\item During the first coherence interval, the transmitted signal is the same as that used to achieve~\eqref{eq:ach_3users_points_bc_4_2}. Therefore, receiver~3 estimates its channel between $\min\left\{ \N_1^{\ast}, \N_3 \right\}$ transmit antennas, and furthermore, receiver~1, and receiver~2 achieve $\frac{\T_2}{\T_1} \N_1^{\ast}\left( \T_1 - \N_1^{\ast}\right)$, and $\left( \frac{\T_2}{\T_1} - 1\right)\N_1^{\ast}\min\left\{ \N_1^{\ast}, \N_2 \right\}$ degrees of freedom, respectively. 
	\item During the remaining coherence intervals, the transmitter sends, every $\T_2$-length interval, the same signal that achieves~\eqref{eq:ach_3users_points_bc_4_2} after multiplying it from the left by $\W_i$ which contains data for receiver~3. Therefore, during the first $\T_1$ of every $\T_2$-length subinterval, the transmitted signal is
	\begin{equation}
          \X^{(123)}= \left[ \W_i, \; \W_i\U_i\right].
  \end{equation}
	After that during $\left( \frac{\T_2}{\T_1} - 1\right)\T_1$ slots, the transmitted signal is 
	\begin{equation}
          \widetilde{\X}^{(123)}= \left[ \W_i\V_i, \; \W_i\V_i\U_i\right].
  \end{equation}
  Thus, receiver~1 can estimate the equivalent channel $\overline{\overline{\H}}_{1,i}=\H_{1,i}\W_i\V_i$, and decode $\U_i$. Also, receiver~2 can estimate the equivalent channel $\overline{\H}_{2,i}=\H_{2,i}\W_i$, and decode $\V_i$ and furthermore, receiver~3 can decode $\W_i$, achieving, respectively, $\left( \frac{\T_3}{\T_2} - 1 \right)\frac{\T_2}{\T_1}\N_1^{\ast}\left( \T_1 - \N_1^{\ast}\right)$, $\left( \frac{\T_3}{\T_2} - 1\right)\left( \frac{\T_2}{\T_1} - 1\right)\N_1^{\ast}\min\left\{ \N_1^{\ast}, \N_2 \right\}$, and $\left( \frac{\T_3}{\T_2} - 1\right)\N_1^{\ast}\min\left\{ \N_1^{\ast}, \N_3 \right\}$ degrees of freedom. 
\end{itemize}
\subsection{Achievability for $\K$ Receivers}
\label{section:ach_kusers_bc}

To obtain the achievability for the $\K$-receiver case, we show that for every set of receivers $\Jset \subseteq \left[ 1:\K \right]$, ordered ascendingly according to the coherence times length, the degrees of freedom tuples $\Dset_1\left(\Jset\right)$ and $\Dset_2\left(\Jset\right)$ are achievable. We use an induction argument in our proof as follows. The achievability when $\Jset$ has only three receivers was shown in Section~\ref{section:nonidentical_ach_proof_bc}. The remainder of the proof is dedicated to show that for arbitrary set of receivers, $\Jset\subset \left[ 1:\K \right]$ where the receivers are ordered ascendingly according to the coherence times length, the product superposition achieves the degrees of freedom tuples $\Dset_1\left(\Jset\right)/\Dset_2\left(\Jset\right)$, we can achieve the degrees of freedom tuple $\Dset_1(\widetilde{\Jset})/\Dset_2(\widetilde{\Jset})$, where $\widetilde{\Jset}\subseteq \left[ 1:\K \right]$ is the set constructed by adding one more receiver to the set $\Jset$ where the length of the added receiver coherence time is an integer multiple of $j_{\max}$. To complete the proof we need to show that product superposition achieves the degrees of freedom tuples $\Dset_1(\widetilde{\Jset})/\Dset_2(\widetilde{\Jset})$ for the set $\widetilde{\Jset}$. The following Lemma addresses this part of the proof.

\begin{lemma}\label{lemma:nonidentical_ach_bc}
For the broadcast channel considered in Section~\ref{section:nonidentical_bc}, define $\X_o \in \mathbb{C}^{\T_{\tau} \times \T_o}$ to be a pilot-based transmitted signal during $\T_o$ slots where a training matrix $\X_{\tau} \in \mathbb{C}^{\T_{\tau} \times \T_{\tau}}$ is sent during $\T_{\tau}$ slots and then the data is sent during $(\T_o-\T_{\tau})$ slots achieving the degrees of freedom tuple $\mathcal{D}^{\text{(o)}} =\left( d_1^{\text{(o)}}, d_2^{\text{(o)}}, \cdots, d_\J^{\text{(o)}}\right)$ for $\J$ receivers. We are able to achieve $\mathcal{D}^{\text{(o)}}$ for the $\J$ receivers in addition to $\left(\frac{\T_{\epsilon}}{\T_o}-1\right)\frac{\T_{\tau}\min\left\{\T_{\tau}, \N_{\epsilon}^{\ast}\right\}}{\T_{\epsilon}}$ to a receiver $\epsilon$ with $\T_{\epsilon}$-length coherence time and $\N_{\epsilon}$ receive antennas, where $\frac{\T_\epsilon}{\T_o} \in {\mathbb Z}$.
\end{lemma}

\begin{proof}
This can be achieved by the following product superposition transmission over $\frac{\T_{\epsilon}}{\T_o}$ coherence intervals of length $\T_o$ slots.
\begin{itemize}
	\item During the first coherence interval, the transmitted signal is 
	      $\X_o$. Thus, $\mathcal{D}^{\text{(o)}}$ degrees of freedom tuple is achieved for the $\J$ receivers and no degrees of freedom for receiver~$\epsilon$, yet, it estimates its channel between $\min\left\{\N_{\epsilon}^{\ast}, \T_{\tau} \right\}$ transmission antennas.
	\item During the remaining coherence intervals, every $\T_o$ slots, the transmitter sends
  \begin{equation}
           \widetilde{\X}_o=\P\X_o,
  \end{equation}
  where $\P \in \mathbb{C}^{\M \times \T_{\tau}}$ contains data for receiver~$\epsilon$. $\X_o$ contains the training matrix $\X_{\tau}$, hence, receiver~$\epsilon$ can decode $\P$, using its channel estimate. Furthermore, the $\J$ receivers estimate their equivalent channels and decode their data during $(\T_o - \T_{\tau})$ slots. Thus, $\J$ receivers achieve $\left( \frac{\T_{\epsilon}}{\T_o} - 1\right)\mathcal{D}^{\text{(o)}}$ degrees of freedom tuple, and furthermore, receiver~$\epsilon$ achieves $\left( \frac{\T_{\epsilon}}{\T_o} - 1\right) \T_{\tau} \min\left\{\N_{\epsilon}^{\ast}, \T_{\tau}\right\}$ degrees of freedom.
\end{itemize}
Thus, in $\T_{\epsilon}$ slots, $\J$ receivers achieve $\frac{\T_{\epsilon}}{\T_o} \mathcal{D}^{\text{(o)}}$ degrees of freedom, and furthermore, receiver~$\epsilon$ achieves $\left( \frac{\T_{\epsilon}}{\T_o} - 1\right) \T_{\tau} \min\left\{\N_{\epsilon}^{\ast}, \T_{\tau}\right\}$ degrees of freedom which completes the proof of Lemma~\ref{lemma:nonidentical_ach_bc}.
\end{proof}
Using Lemma~\ref{lemma:nonidentical_ach_bc} the second part of the proof is completed, and hence, the proof of Theorem~\ref{theorem:nonidentical_ach_bc} is completed. 

\section{General Coherence Times}
\label{section:general_coherence}

In this section, we study a $\K$-receiver broadcast channel with general coherence times. An achievable degrees of freedom region is obtained, where the coherence times have arbitrary ratio or alignment.\footnote{Coherence times, as is required in a block fading model in a time-sampled domain, continue to take positive integer values.} 

\subsection{Unaligned Coherence Times}
\label{section:unaligned_coherence}
\begin{figure}
\center
\includegraphics[width=\Figwidth]{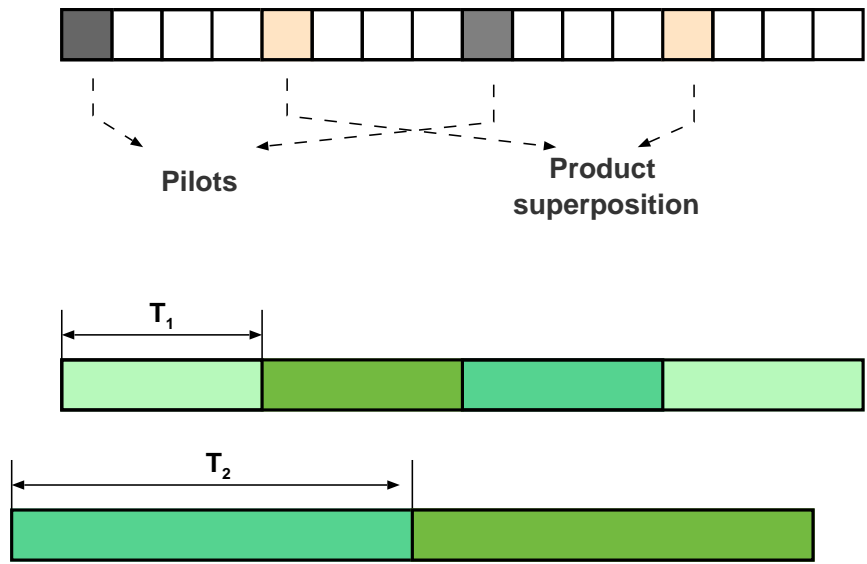}
\caption{Product superposition transmission for unaligned coherence times, where $\T_2 = 2 \T_1 = 6$.}
\label{fig:PS_unaligned}
\end{figure}

\begin{figure}
\center
\includegraphics[width=\Figwidth]{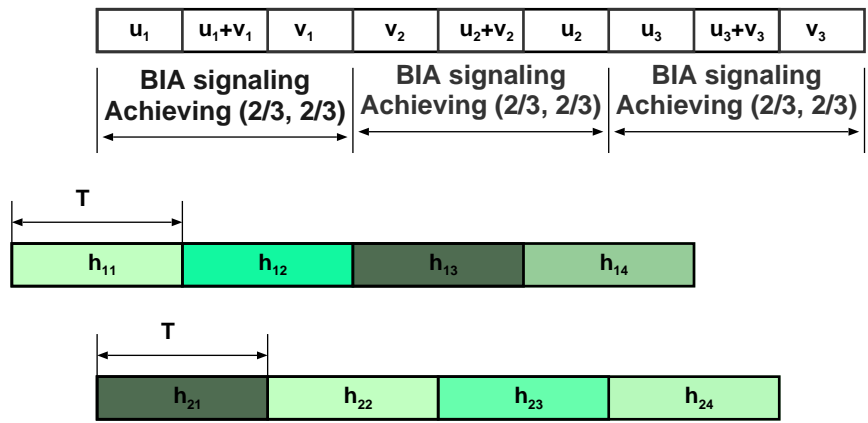}
\caption{Blind interference alignment for staggered coherence times with CSIR, where $\T_1=\T_2=2$. Receiver~1 cancels $\h_{1i}^H \v$, and decodes $\u$, whereas receiver~2 cancels $\h_{2i}^H \u$, and decodes $\v$ achieving $\left( \frac{2}{3}, \frac{2}{3} \right)$ degrees of freedom pair.}
\label{fig:BIA_CSIR}
\end{figure}
In this section, we relax the assumption on the alignment of coherence intervals.
Consider a broadcast channel with $\K$ receivers where the coherence times are integer multiple of each other, i.e. $\frac{\T_k}{\T_{k-1}} \in {\mathbb Z}$. The coherence times have arbitrary alignment, meaning that there could be an offset between the transition times of the coherence intervals of different receivers. Recall that in the case of aligned coherence intervals, product superposition provided the achievable degrees of freedom  region in~\eqref{eq:ach_bc_conv_hull}. The receiver with longer coherence time reuses some of the unneeded pilots and achieves gains in degrees of freedom without affecting the receivers with shorter coherence times. Under unaligned coherence times the same gains in degrees of freedom are available with product superposition. Using the transmitted signal given in Section~\ref{section:nonidentical_ach_proof_bc}, the longer coherence times include the same number of unneeded pilot sequences regardless of the alignment. These unneeded pilot sequences can be reused by product superposition transmission, achieving degrees of freedom gain. For instance, consider two receivers with $\M = 2, \N_1 = \N_2 = 1, \T_1 = 4, \T_2 = 8$, with an offset of one transmission symbol as shown in Fig.~\ref{fig:PS_unaligned}. We can achieve the degrees of freedom pair $(\frac{3}{4}, \frac{1}{8})$ via a transmission strategy over pairs of coherence intervals for receiver~1, as follows.
\begin{itemize}
	\item In the odd coherence intervals for receiver~1, one pilot is transmitted during which both receivers estimate their channels. In the  3 remaining time slots of this interval, data is transmitted for receiver~1.
	\item In the even coherence intervals, during the first time slot a product superposition is transmitted providing one degree of freedom for receiver~2 (whose channel has not changed) while allowing receiver~1 to renew the estimate of his channel. The three remaining time slots provide 3 further degrees of freedom for receiver~1.
\end{itemize}
Thus, in 8 time slots, receiver~1 achieves 6 degrees of freedom  and receiver~2 achieves 1. This is the same ``corner point'' that is obtained in the aligned scenario, noting that the nature of the algorithm is not changed, only the position of the pilot transmission must be carefully chosen while keeping in mind the transition points of the block fading.
\subsection{Unaligned Coherence Times with Perfect Symmetry (Staggered)}
\label{section:staggered_coherence}

We now consider a special case of two-receiver unaligned coherence times where the transition of each coherence interval is exactly in the middle of the other coherence interval. This special case is motivated by the blind interference alignment model in Fig.~\ref{fig:BIA_CSIR} that was considered in \cite{Jafar_blind}, and for easy reference we call this configuration a {\em staggered} coherence times.
 
 We follow the example of blind interference alignment~\cite{Jafar_blind}: a 2-receiver broadcast channel with $\M=2, \N_1 = \N_2 =1$. As shown in Fig.~\ref{fig:BIA_pilots}, the transitions of the longer coherence interval occur at the middle of the shorter coherence interval. Based on the discussion in Section~\ref{section:unaligned_coherence}, product superposition can obtain degrees of freedom gain for the staggered scenario. In~\cite{Jafar_blind}, blind interference alignment achieved degrees of freedom pair $\left( \frac{2}{3}, \frac{2}{3}\right)$ while ignoring the cost of CSIR, which is a key part of our analysis. To allow comparison and synergy, we analyze a version of blind interference alignment with channel estimation shown in Fig.~\ref{fig:BIA_pilots}. The gain of blind interference alignment comes from the staggering of the coherence time, whereas the source of product superposition gain is reusing the unneeded pilots with respect to the longer coherence times. Therefore, we can give a transmission scheme that uses both blind interference alignment and product superposition over $\frac{\T_2}{\T_1}$ coherence intervals of receiver~1, as shown in Fig.~\ref{fig:BIA_PS}.
\begin{itemize}
	\item During the first coherence interval, two pilots are sent in the middle of the interval. Receiver~1 estimates its channel during this interval, whereas receiver~2 estimates its channel as these two pilots are located at the first, and the last time slots of its coherence interval.
	\item Blind interference alignment signaling is sent during the remaining $\left( \T_1 - 2 \right)$ time slots of first interval and the first $\frac{1}{2} \left( \T_1 - 2 \right)$ time slots of the second interval. Hence, the degrees of freedom pair $\left( \left( \T_1 - 2\right), \left( \T_1 - 2\right)\right)$ is achieved.
	\item Product superposition signaling is sent during the remaining $\frac{1}{2} \left( \T_1 - 2 \right)$ time slots of the second interval. Receiver~1 estimates its channel of the second coherence interval, and furthermore, achieves further $\frac{1}{2} \left( \T_1 - 2 \right)$, whereas, receiver~2 achieves $2$ further degrees of freedom.
	\item Furthermore, during the remaining $\left( \frac{\T_2}{\T_1} - 2\right)$ receiver~1 coherence intervals, product superposition signaling is sent achieving the degrees of freedom pair
	\begin{equation}
    \left( \left( \frac{\T_2}{\T_1} -2\right) \left( \T_1 - 1 \right), \left( \frac{\T_2}{\T_1} -2\right) \right)
	\end{equation}
	is achieved.
\end{itemize}
Thus, the above transmission scheme obtain the degrees of freedom pair 
\begin{equation}
\left( 1 - \frac{1}{\T_1} - \frac{1}{\T_2} - \frac{\T_1}{2\T_2}, \frac{\T_1}{\T_2} + \frac{1}{\T_1} - \frac{2}{\T_2} \right).
\end{equation}
Furthermore, product superposition transmission only can achieve the degrees of freedom pair $\left( 1 - \frac{1}{\T_1}, \frac{1}{\T_2} - \frac{1}{\T_1} \right)$. Hence, the achievable degrees of freedom is the convex hull of the degrees of freedom pairs achieved by blind interference alignment, product superposition, and combining blind interference alignment with product superposition.
\begin{figure}
\center
\includegraphics[width=\Figwidth]{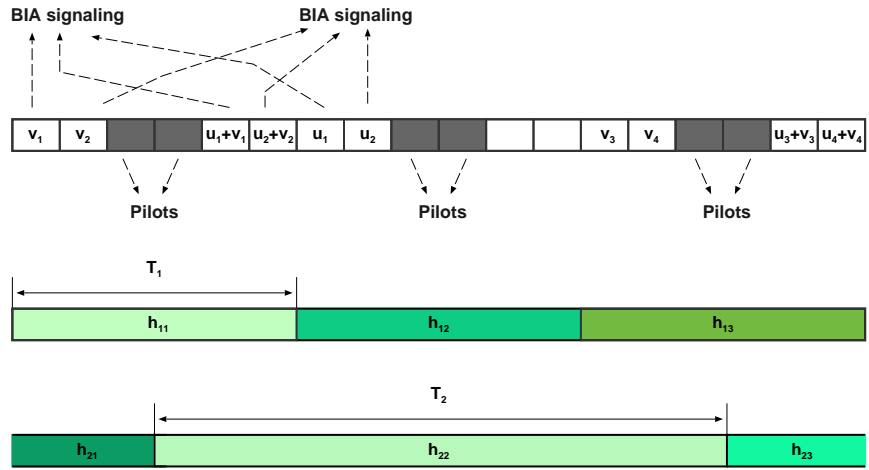}
\caption{Blind interference alignment with pilot transmission, where $\T_2=2\T_1=12$.}
\label{fig:BIA_pilots}
\end{figure}

\begin{figure}
\center
\includegraphics[width=\Figwidth]{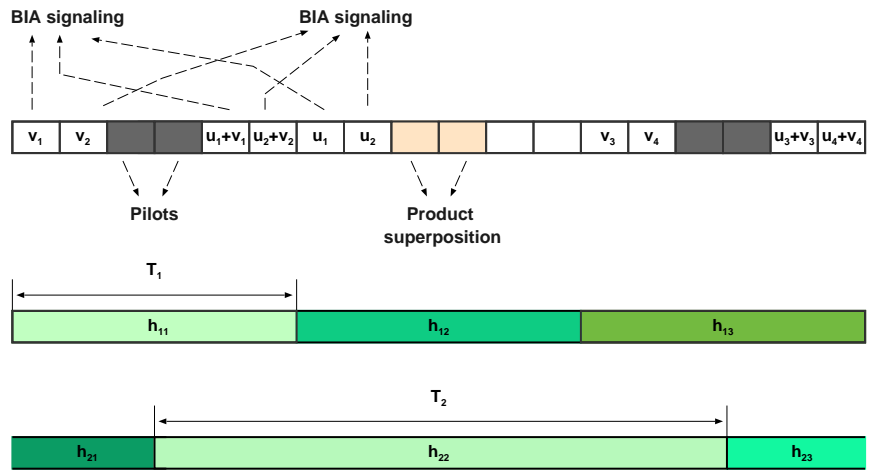}
\caption{Combining blind interference alignment with product superposition, where $\T_2=2\T_1=12$.}
\label{fig:BIA_PS}
\end{figure}

\subsection{Arbitrary Coherence Times}
\label{section:rational_coherence}
\begin{theorem}
\label{theorem:rational_ach_bc}
Consider a $\K$-receiver broadcast channel without CSIT or CSIR having heterogeneous coherence times, where the coherence times are allowed to take any positive integer value. Product superposition can achieve the degrees of freedom tuple defined in~\eqref{eq:nonidentical_ach_bc_2}.
\end{theorem}

\begin{remark}
Blind interference alignment signaling can be sent at the location of the staggering coherence times achieving degrees of freedom gain. Hence, similar to the case of staggered coherence times with integer ratio in Section~\ref{section:staggered_coherence}, product superposition can be combined with blind interference alignment increasing the achievable degrees of freedom region.
\end{remark}

\begin{proof}
For clarity of explanation, we start by giving the achievable scheme for 3 receivers with $\N_k = \N \leq \min\left\{ \M, \left\lfloor \frac{\T_1}{2} \right \rfloor \right\}, \forall k$ over $\T_2 \T_3$ coherence intervals of receiver~1.
\begin{itemize}
	\item For every coherence interval, a pilot sequence of length $\N$ slots, and receiver~1 data of length $\T_1 - \N$ slots are sent, achieving $\N \left( \T_1 - \N \right)$ degrees of freedom for receiver~1.
 \item The number of pilot sequences of length $\N$ is $\T_2 \T_3$. Having coherence time $\T_2$, receiver~2 needs only $\T_1 \T_3$ pilot sequences for channel estimation. Hence, produced superposition can be sent during $( \T_2 \T_3 - \T_1 \T_3 )$ pilot sequences to send data for receiver~2 achieving $\N \T_3 (\T_2 - \T_1)$ degrees of freedom.
 \item Furthermore, receiver~3 needs only $\T_1 \T_2$ pilot sequences for channel estimation, and hence, data signal for receiver~3 can be sent during $( \T_2 \T_3 - \T_1 \T_2)$ pilot sequences via product superposition. Product superposition uses $( \T_2 \T_3 - \T_1 \T_3 )$ pilot sequences to send data for receiver~2, and hence, receiver~3 can only reuse $( \T_2 \T_3 - \T_1 \T_2) - ( \T_2 \T_3 - \T_1 \T_3 ) = \T_1( \T_3- \T_2 )$ pilot sequences achieving $\N \T_1 (\T_3 - \T_2)$ degrees of freedom.
\end{itemize}
Thus, normalized over $\T_1\T_2\T_3$ time slots, we can achieve the degrees of freedom tuple
\begin{equation}
\left( \N \left( 1 - \frac{\N}{\T_1} \right), \N^2 \left( \frac{1}{\T_1} - \frac{1}{\T_2} \right), \N^2 \left( \frac{1}{\T_2} - \frac{1}{\T_3} \right) \right).
\end{equation}

Now, we give the proof for arbitrary number of receivers, and general antenna setup. For a set of receiver $\Jset \subseteq \left[ 1:\K \right]$ having $\J$ receiver where, $\frac{\T_j}{\T_{j-1}} \in \Qset, \ j \in \Jset$, the degrees of freedom tuple~\eqref{eq:nonidentical_ach_bc_2} can be obtained over $ \prod_{i=2}^\J \T_i$ coherence intervals of receiver~$j_{\min}$.
\begin{itemize}
	\item For every interval, a pilot sequence of length $\N_{j_{\min}}^{\ast}$ slots, and data of length 
				$\left( \T_{j_{\min}} - \N_{j_{\min}}^{\ast}\right)$ for receiver~$j_{\min}$ are sent, achieving $\N_{j_{\min}}^{\ast} \left( \T_{j_{\min}} - \N_{j_{\min}}^{\ast} \right)$ degrees of freedom.
 \item The number of pilot sequences of length $\N_{j_{\min}}^{\ast}$ slots is $\prod_{i=2}^\J \T_i$. Receiver~$j \neq j_{\min}$, with coherence time $\T_j$, can estimate the channel of $\min \left\{ \N_{j_{\min}}^{\ast}, \N_j \right\}$ 
       transmit antennas using $\prod_{i=1, i\neq j }^\J \T_i$ pilot sequences. Excluding the pilots reused by receivers $\left\{ j_{\min}+1, \cdots, j-1 \right\}$ to send data by product superposition transmission, data for receiver~$j$ can be sent via product superposition during $ (\T_j - \T_{j-1} ) \prod_{i=1, i\notin \{j, j-1\}}^\J \T_i$ pilots obtaining the 
			 degrees of freedom $\N_{j_{\min}}^{\ast} \min \left\{ \N_{j_{\min}}^{\ast}, \N_j \right\} (\T_j - \T_{j-1} ) \prod_{i=1, i\notin \{j, j-1\}}^\J \T_i$.
\end{itemize}
Thus, the proof of Theorem~\ref{theorem:rational_ach_bc} is completed.
\end{proof}

\section{Multiple Access Channel with Identical Coherence Times}
\label{section:identical_mac}

Consider a $\K$-transmitter MIMO multiple access channel without CSIT or CSIR, where transmitter~$k$ is equipped with $\M_k$ antennas, and the receiver is equipped with $\N$ antennas. The received signal at the discrete time $n$ can be given by
\begin{equation}\label{eq:received_mac}
\y(n) = \sum_{k=1}^\K \overline{\H}_k(n) \x_k(n) + \z(n), 
\end{equation}
where $\x_k(n) \in \mathbb{C}^{\M_k \times 1}$ is transmitter~$k$ signal, $\z(n) \in \mathbb{C}^{\N \times 1}$ is the i.i.d. Gaussian additive noise and $\overline{\H}_k(n) \in \mathbb{C}^{\N \times \M_k}$ is transmitter~$k$ Rayleigh block-fading channel matrix with coherence time $\T_k$~\cite{Marzetta_capacity}. We study the case when $\T_k \geq 2\N, \forall k$~\cite{Zheng_communication}.

Assume that all transmitters have identical coherence times, $\T$. In the sequel, we define a degrees of freedom achievable region based on a pilot-based scheme in Section~\ref{section:identical_ach_mac}. Furthermore, an outer degrees of freedom region is given in Section~\ref{section:identical_outer_mac} based on the cooperative bound. Some numerical examples are given in Section~\ref{section:identical_examples_mac} where it is shown that the achievable degrees of freedom region is tight against sum degrees of freedom.

\subsection{Achievability}
\label{section:identical_ach_mac}

\begin{theorem}
\label{theorem:identical_ach_mac}
Consider a $K$-transmitter MIMO multiple access channel without CSIT or CSIR, meaning that the channel realization is not known, but the channel distribution is globally known. If the transmitters have identical coherence times, namely $\T$, then for every ordered set of transmitters, $\Jset=\left\{ k_1, k_2, \cdots, k_J\right\} \subseteq \left[ 1:\K \right]$, we can achieve the set of degrees of freedom tuples $\Dset(\Jset):$
\begin{equation}\label{eq:identical_ach_mac}
d_j = \M'_j \left( 1 - \frac{\sum_{j \in \Jset}{\M'_j}}{\T}\right), \quad j \in \Jset,
\end{equation}
where $\M'_j = \min{\left\{ \M_j , \left[ N - \sum_{m=1}^{j-1}\M'_{k_m}\right]^+\right\}}$, and $\T \geq 2\N$. The achievable degrees of freedom region is the convex hull of the degrees of freedom tuples, $\Dset(\Jset)$, over all the $\sum_{i=1}^\K \frac{\K!}{(\K-i)!}$ possible ordered sets $\Jset \subseteq \left[ 1: \K\right]$, i.e.,
\begin{align}
\mathcal{D} = \Big\{ \left( \d_1, \cdots, \d_\K \right) \in \text{Co} \left( \Dset(\Jset) \right), 
                      \forall {\Jset \subseteq \left[ 1: \K \right]} \Big\}.
\end{align}
\end{theorem}

\begin{proof}
We show that a simple pilot-based scheme can achieve the above achievable degrees of freedom region. Assume that we have an ordered set of transmitters $\Jset=\left\{ k_1, \cdots, k_J\right\} \subseteq \left[ 1:\K \right]$. In order to achieve the degrees of freedom tuple in~\eqref{eq:identical_ach_mac}, we can use the following transmission scheme over coherence interval of length $\T$ time slots. 
\begin{itemize}
  \item During the first $\sum_{j \in \Jset}{\M'_j}$ slots of the coherence interval, a pilot sequence is sent, where transmitter~$j$ sends $\M'_j$ pilots. The receiver estimates the channel of the corresponding $\sum_{j \in \Jset}{\M'_j}$ transmit antennas.
	\item During the remaining $\left(\T-\sum_{j \in \Jset}{\M'_j}\right)$ slots, simultaneously, $\M'_j \left( \T - \sum_{j \in \Jset}{\M'_j} \right)$ data matrix is sent from transmitter~$j$.	Hence, the receiver, using $\sum_{j \in \Jset}{\M'_j}$ antennas, can invert the channel and decode the transmitted signal.
\end{itemize}
Therefore, every $\T$ period, transmitter~$j \in \Jset$ can achieve $\M'_j \left( \T - \sum_{j \in \Jset}{\M'_j} \right)$ degrees of freedom, and hence~\eqref{eq:identical_ach_mac} is obtained.
\end{proof}

\subsection{Outer Bound}
\label{section:identical_outer_mac}

For the considered $\K$-transmitter multiple access channel with identical coherence times, namely $\T$, the cooperative bound~\cite{Sato_outer} can be given by~\cite{Gamal_network}
\begin{equation}\label{eq:identical_outer_mac}
\sum_{j \in \Jset} \R_j \leq I \left( X\left( \Jset \right); Y | X\left(\Jset^c\right) \right), \quad \forall 
\Jset \subseteq \left[ 1:\K \right].
\end{equation}
An outer bound on the degrees of freedom region is~\cite{Zheng_communication},
\begin{align}\label{eq:identical_outer_mac_1}
\sum_{j \in \Jset} & \d_j \leq \min\left\{\N, \sum_{j \in \Jset} \M_j \right\} \left( 1 - \frac{\min\left\{ \N, \sum_{j \in \Jset} \M_j \right\}}{\T}\right), \quad \forall \Jset \subseteq \left[ 1:\K \right].
\end{align}

\begin{figure}
\center
\includegraphics[width=\Figwidth]{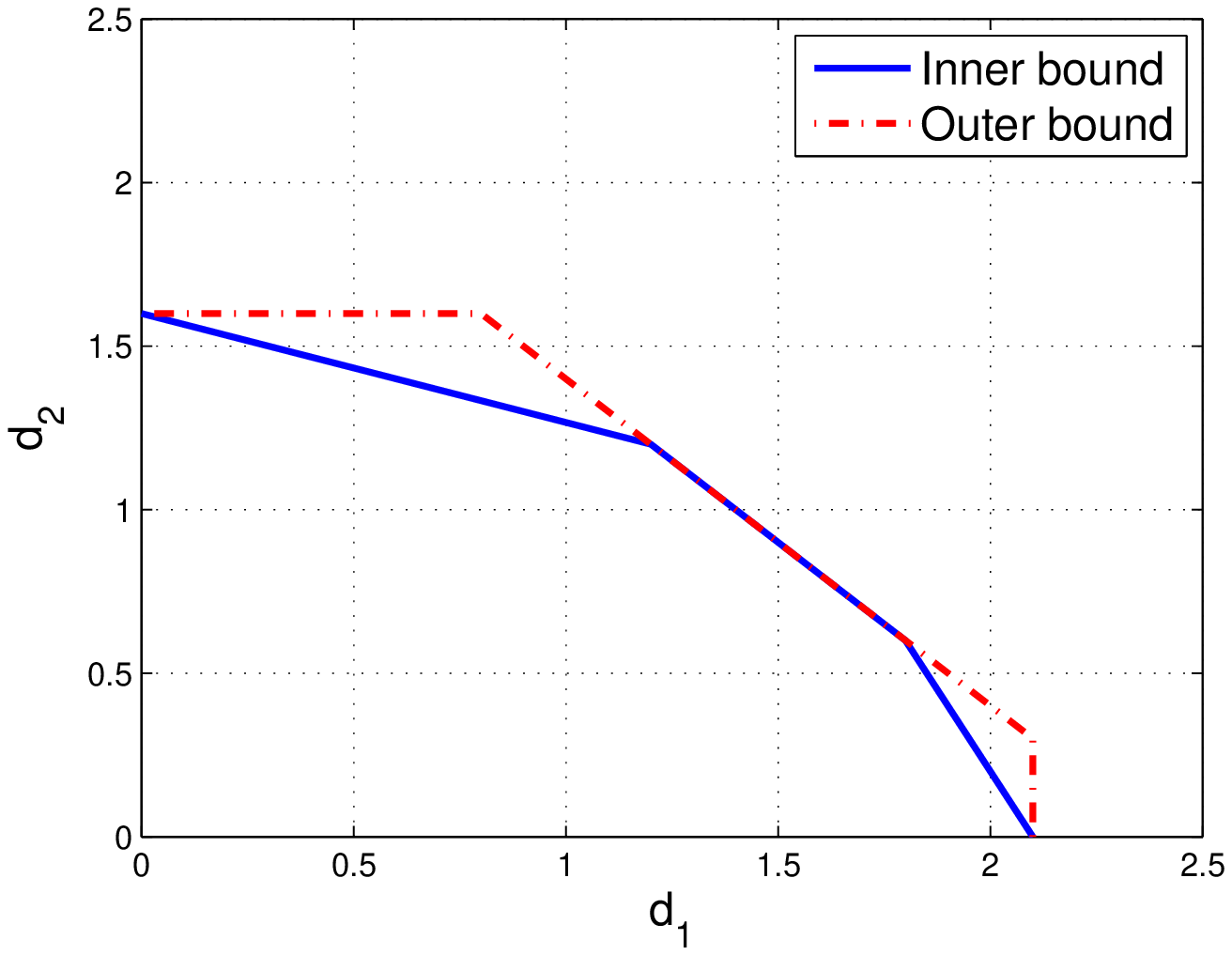}
\caption{Degrees of freedom region of a two-transmitter multiple access channel with identical coherence times $\T=10$, and $\M_1=3, \M_2=2, \N=4$.}
\label{fig:identical_2users_mac_1}
\end{figure}

\subsection{Numerical Examples}
\label{section:identical_examples_mac}

\begin{figure}
\center
\includegraphics[width=\Figwidth]{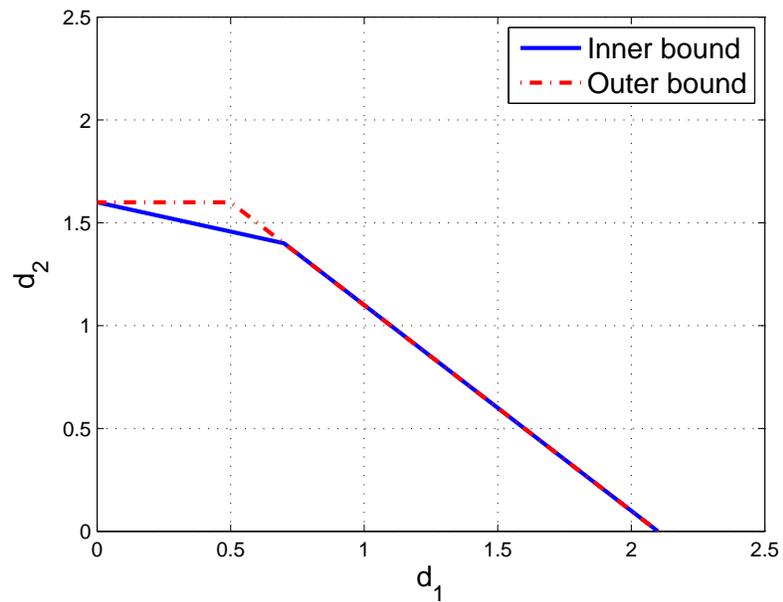}
\caption{Degrees of freedom region of a two-transmitter multiple access channel with identical coherence times $\T =10$, and $\M_1=4, \M_2=2, \N=3$.}
\label{fig:identical_2users_mac_2}
\end{figure}

Consider a two-transmitter multiple access channel with $\M_1=3, \M_2=2, \N=4, \T=10$. The outer degrees of freedom region is given by 
\begin{align*}
      \d_1        &\leq \frac{21}{10}, \nonumber\\
      \d_2        &\leq \frac{16}{10}, \nonumber\\
      \d_1 + \d_2 &\leq \frac{24}{10}.
\end{align*}
The achievable degrees of freedom pairs in Theorem~\ref{theorem:identical_ach_mac} can be obtained as follows. For the case of two transmitters, there are 5 ordered sets of transmitters $\Jset$: $\left\{ \right\}, \left\{1\right\}, \left\{2\right\}$, $\left\{1, 2\right\}$ and $\left\{2, 1\right\}$. For $\left\{ \right\}$, the trivial degrees of freedom pair $\left(0,0\right)$ can be obtained. For the two sets $\left\{1\right\}, \left\{2\right\}$, the degrees of freedom pairs $\left(\frac{21}{10},0\right)$ and $\left(0,\frac{16}{10}\right)$, respectively, can be obtained. For the two sets $\left\{1, 2\right\}$ and $\left\{2, 1\right\}$, the degrees of freedom pairs $\left(\frac{18}{10},\frac{6}{10}\right)$ and $\left(\frac{12}{10},\frac{12}{10}\right)$, respectively, can be obtained. The convex hull of the achieved degrees of freedom pairs gives the achievable degrees of freedom region which is tight against the sum degrees of freedom as shown in Fig~\ref{fig:identical_2users_mac_1}.

Consider a two-transmitter multiple access channel with $\M_1=4, \M_2=2, \N=3, \T=10$. As shown in Fig.~\ref{fig:identical_2users_mac_2}, the achievable degrees of freedom regions are tight against the sum degrees of freedom.

\section{Multiple Access Channel with Heterogeneous Coherence Times}
\label{section:nonidentical_mac}

Consider the multiple access channel defined in~\eqref{eq:received_mac} where there is no CSIT or CSIR. Consider the case where the receivers coherence times are perfectly aligned and integer multiples of each others, i.e., $\forall k, \frac{\T_k}{\T_{k-1}}\in {\mathbb Z}$. In the sequel, we give an achievable, and an outer degrees of freedom regions in Section~\ref{section:nonidentical_ach_mac} and Section~\ref{section:nonidentical_outer_mac}, respectively. Furthermore, some numerical examples are given in Section~\ref{section:nonidentical_examples_mac} to demonstrate the achievable and the outer degrees of freedom regions.

\subsection{Achievability}
\label{section:nonidentical_ach_mac}

\begin{theorem}
\label{theorem:nonidentical_ach_mac}
Consider a $K$-transmitter MIMO multiple access channel without CSIT or CSIR, meaning that the channel realization is not known, but the channel distribution is globally known. Furthermore, the transmitters coherence times are assumed to be perfectly aligned and integer multiples of each other. Define $\Jset=\left\{ i_1, \cdots, i_J\right\} \subseteq \left[ 1:\K \right]$ to be a set of $\J$ transmitters where $\forall j \in \Jset, \frac{\T_j}{\T_{j-1}}\in {\mathbb Z}$. Define $\breve{\Jset}=\left\{ k_1, \cdots, k_J\right\}$ to be one of the $\J!$ possible ordered sets of $\Jset$. If $\T_k \geq 2\N, \forall k$, we can achieve the set of degrees of freedom tuples $\Dset(\breve{\Jset}):$
\begin{align}\label{eq:nonidentical_ach_mac}
d_j = \M'_j \sum_{m=1}^\J{\left( \T_{i_1} - \sum_{n=1}^m{\M'_{i_n}}\right)\left( \frac{1}{\T_{i_m}} - \frac{1}{\T_{i_{m+1}}} \right)},
\end{align}
where $\M'_j = \min{\left\{ \M_j , \left[ N - \sum_{m=1}^{j-1}\M'_{k_m}\right]^+\right\}}$, and, for notational convenience, we introduce the trivial random variable $\T_{i_{\J+1}}$, i.e.,  $\frac{1}{\T_{i_{\J+1}}}=0$. Hence, the achievable degrees of freedom region is the convex hull of the degrees of freedom tuples, $\Dset(\Jset)$, over all the $\sum_{i=1}^\K \frac{\K!}{(\K-i)!}$ possible ordered sets $\breve{\Jset} \subseteq \left[ 1: \K\right]$, i.e.,
\begin{align}
\mathcal{D} = \Big\{ \left( \d_1, \cdots, \d_\K \right) \in \text{Co} \left( \Dset(\breve{\Jset}) \right), 
                      \forall {\breve{\Jset} \subseteq \left[ 1: \K \right]} \Big\}.
\end{align}
\end{theorem}

\begin{proof}
By time-sharing between the transmission schemes that achieve the degrees of freedom tuples $\Dset(\breve{\Jset})$, we can construct the achievable degrees of freedom region which is the convex hull of the achieved degrees of freedom tuples. The remainder of the proof is dedicated to show the achievability of the degrees of freedom tuple in~\eqref{eq:nonidentical_ach_mac} using the following transmission scheme over $\frac{\T_{i_\J}}{\T_{i_1}}$ coherence intervals of transmitter $i_1$.
\begin{itemize}
	\item During the first coherence interval, $\sum_{j\in \Jset}^\J \M'_j$ pilots are sent to estimate $\M'_j$ antennas of transmitter~$j$, and hence, during the following $\left(\T_{i_1}-\sum_{j\in \Jset}^\J \M'_j\right)$ time slots, the transmitters can communicate coherently achieving $\M'_j \left( \T_{i_1}-\sum_{j\in \Jset}^\J \M'_j\right)$ degrees of freedom for transmitters~$j \in \Jset$.
	\item During the remaining intervals, when the index of the interval is $\l \frac{\T_{i_m}}{\T_{i_1}}+1$, where $m=2, \cdots \J-1$ and $\l=1, \cdots, \frac{\T_{i_m+1}}{\T_{i_m}}-1, \frac{\T_{i_m+1}}{\T_{i_m}}+1, \cdots, 2\frac{\T_{i_m+1}}{\T_{i_m}}-1, \frac{\T_{i_m+1}}{\T_{i_m}}+1, \cdots, \frac{\T_{i_\J}}{\T_{i_m}}-1$, the channel of transmitter~$j=i_1, \cdots, i_m$ needs to be estimated, whereas the channel of transmitter~$j=i_m+1, \cdots, \J$ stays the same. Hence, $\sum_{n=1}^m \M'_{i_n}$ pilots are sent to estimate $\M'_j$ antennas of transmitters~$j=i_1, \cdots, i_m$. After that, during the following $\left(\T_{i_1}-\sum_{n=1}^m \M'_{i_n}\right)$ slots, the transmitters can communicate coherently achieving $\M'_j \left( \T_{i_1}-\sum_{n=1}^m \M'_{i_n}\right)$ degrees of freedom for transmitters~$j \in \Jset$. The number of intervals with index $k\frac{\T_m}{\T_{m-1}}+1$ is $\sum_{m=2}^\J \left( \frac{\T_{i_{m+1}}}{\T_{i_m}} - 1 \right)\frac{\T_{i_\J}}{\T_{i_{m+1}}}$. 
  \item For the intervals of length $\T_{i_1}$ with index not equal to $\l \frac{\T_{i_m}}{\T_{i_1}}+1$, the channels of all transmitters remain the same except the channel of transmitter~$i_1$. Hence, $\M'_{i_1}$ pilots are sent to estimate the channel of transmitter~$i_1$, after that the transmitters can communicate coherently during the following $\left(\T_{i_1} - \M'_{i_1}\right)$ slots, achieving $\M'_j \left( \T_{i_1} - \M'_{i_1}\right)$ degrees of freedom for transmitter~$j \in \Jset$. The number of the intervals with index not equal to $\l \frac{\T_{i_m}}{\T_{i_1}}+1$ is $\left( \frac{\T_{i_2}}{\T_{i_1}} - 1 \right)\frac{\T_{i_\J}}{\T_{i_2}}$.
\end{itemize}
Thus, transmitter~$j \in \Jset$ achieves $\M'_j \sum_{m=1}^\J \left( \T_{i_1}-\sum_{n=1}^m \M'_{i_n} \right) \left( \frac{1}{\T_{i_m}} - \frac{1}{\T_{i_{m+1}}} \right)\T_{i_\J}$ degrees of freedom over $\T_{i_\J}$ slots, obtaining~\eqref{eq:nonidentical_ach_mac} which completes the proof of Theorem~\ref{theorem:nonidentical_ach_mac}. 
\end{proof}

\subsection{Outer Bound}
\label{section:nonidentical_outer_mac}

\begin{theorem}
\label{theorem:nonidentical_outer_mac}
Consider a $K$-transmitter MIMO multiple access channel without CSIT or CSIR, meaning that the channel realization is not known, but the channel distribution is globally known. Furthermore, the transmitters coherence times are assumed to be perfectly aligned and integer multiples of each other. Define $\Jset=\left\{ i_1, \cdots, i_\J \right\} \subseteq \left[ 1:\K \right]$ to be a set of $\J$ transmitters where $\frac{\T_j}{\T_{j-1}}\in {\mathbb Z}, \ \T_j \geq 2\N, \ \forall j \in \Jset$. For every $\Jset \subseteq \left[ 1:\K \right]$, if a set of degrees of freedom tuples $\left( \d_{i_1}, \cdots, \d_{i_\J} \right)$ is achievable, then it must satisfy the inequalities
\begin{equation}\label{eq:nonidentical_outer_mac}
\sum_{j \in \Jset} \d_j \leq \min\left\{ \N, \sum_{j \in \Jset} \M_j \right\} \left( 1 - \frac{\min\left\{ \N, \sum_{j \in \Jset} \M_j \right\}}{\T_{i_\J}}\right).
\end{equation}
\end{theorem}

\begin{proof}
The proof is divided into two parts. First, we enhance the channel by increasing the coherence times of the receivers so that the enhanced channel has identical coherence times.
\begin{lemma}\label{lemma:enhance_mac}
For the considered $K$-transmitter MIMO multiple access channel, define $\mathcal{D}\left(\Jset\right)$ to be the degrees of freedom region of a set of transmitters $\Jset= \left\{ i_1, \cdots, i_\J \right\} \subseteq \left[ 1:\K \right]$ with $\frac{\T_j}{\T_{j-1}}\in {\mathbb Z}, \forall j \in \Jset$. Define $\mathcal{\overline{D}}\left(\Jset\right)$ to be the degrees of freedom region of the same set of transmitters $\Jset= \left\{ i_1, \cdots, i_\J \right\} \subseteq \left[ 1:\K \right]$ with $\T_j= \T_{i_J}, \forall j \in \Jset$, where the transmitters have identical coherence times, namely $\T_{i_\J}$. Thus, we have
\begin{equation}
\mathcal{D}\left(\Jset\right) \subseteq \mathcal{\overline{D}}\left(\Jset\right) 
\end{equation}
\end{lemma}
\begin{proof}
See Appendix~\ref{appendix:enhance_mac}.
\end{proof}
Now we show the second part of the proof. The enhanced channel has identical coherence times, namely $\T_{i_\J}$, hence, the cooperative outer bound~\cite{Sato_outer} bound is~\cite{Gamal_network},
\begin{equation}
\sum_{j \in \Jset} \R_j \leq I \left( \X \left(\Jset\right) ; \Y | \X \left( \Jset^c\right)\right).
\end{equation}
According to the results of non-coherent communication in~\cite{Zheng_communication}, the bound in~\eqref{eq:nonidentical_outer_mac} can be obtained, and the proof of Theorem~\ref{theorem:nonidentical_outer_mac} is completed.
\end{proof}

\subsection{Numerical Examples}
\label{section:nonidentical_examples_mac}

\begin{figure}
\center
\includegraphics[width=\Figwidth]{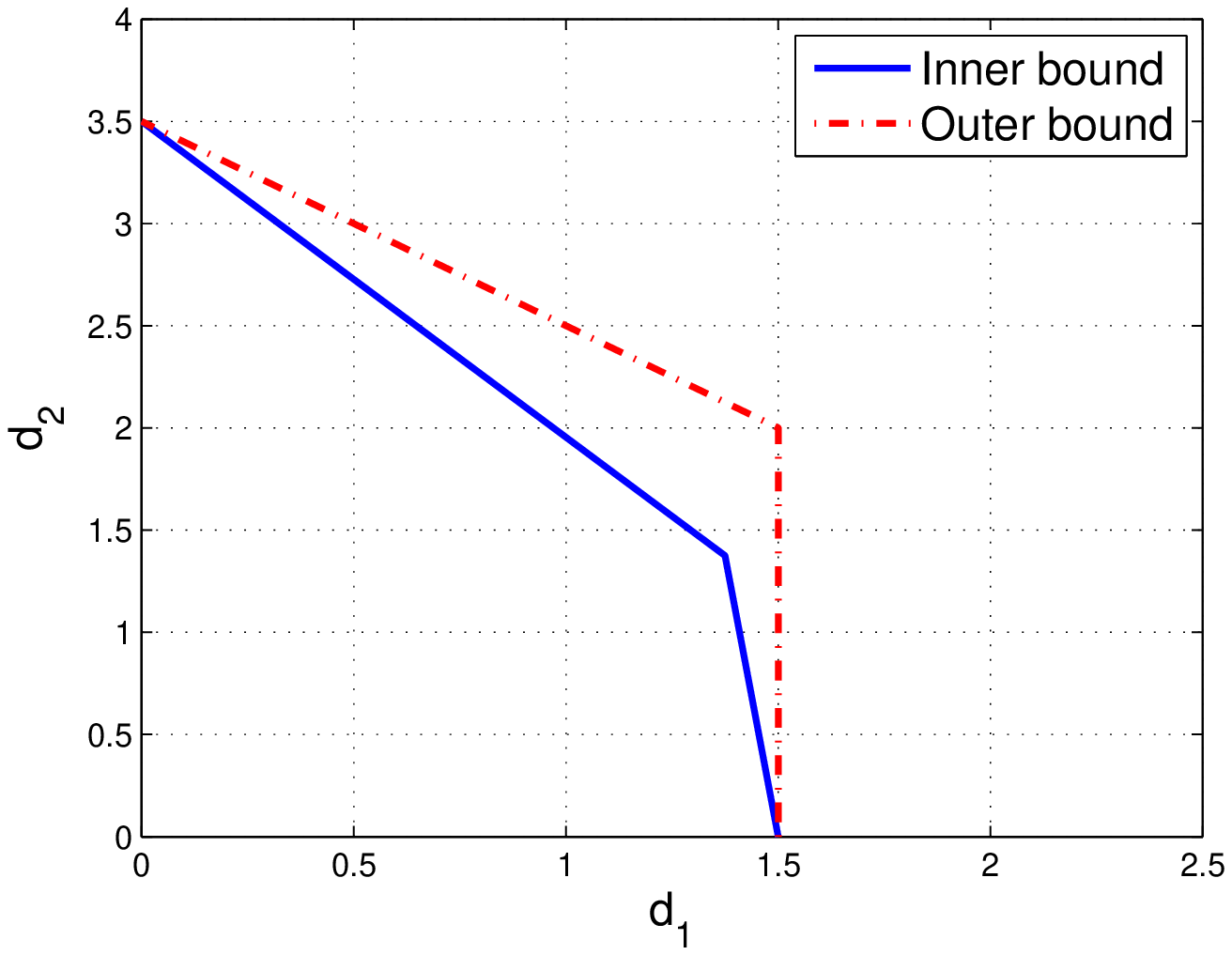}
\caption{Degrees of freedom region of a two-transmitter multiple access channel with heterogeneous coherence times $\T_1=8, \T_2=32$ and $\M_1=2, \M_2=4, \N=4$.}
\label{fig:nonidentical_2users_mac_1}
\end{figure}

\begin{figure}
\center
\includegraphics[width=\Figwidth]{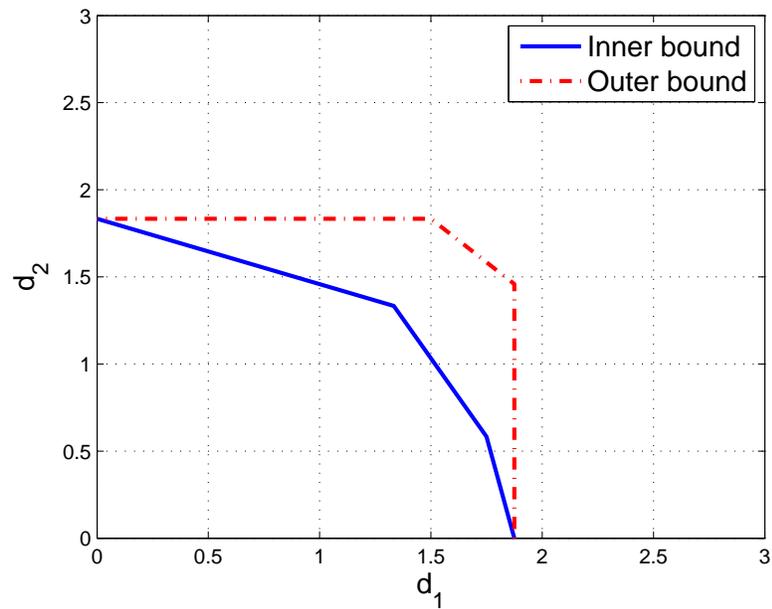}
\caption{Degrees of freedom region of a two-transmitter multiple access channel with heterogeneous coherence times $\T_1=8, \T_2=24$ and $\M_1=3, \M_2=2, \N=4$.}
\label{fig:nonidentical_2users_mac_2}
\end{figure}

Consider a two-transmitter multiple access channel with $\M_1=2, \M_2=\N=4, \T_1=8, \T_2=32$. From Theorem~\ref{theorem:nonidentical_outer_mac}, the outer degrees of freedom region is given by 
\begin{align}
      \d_1        &\leq \frac{3}{2}, \nonumber\\
      \d_1 + \d_2 &\leq \frac{7}{2}.
\end{align}
The achievable degrees of freedom pairs in Theorem~\ref{theorem:nonidentical_ach_mac} can be obtained as follows. There are 5 ordered sets of transmitters $\Jset$: $\left\{ \right\}, \left\{1\right\}, \left\{2\right\}$, $\left\{1, 2\right\}$ and $\left\{2, 1\right\}$. For $\left\{ \right\}$, the trivial degrees of freedom pair $\left(0,0\right)$ can be obtained. For the two sets $\left\{1\right\}, \left\{2\right\}$, the degrees of freedom pairs $\left(\frac{3}{2},0\right)$ and $\left(0,\frac{7}{2}\right)$, respectively, can be obtained. For the two sets $\left\{1, 2\right\}$ and $\left\{2, 1\right\}$, the degrees of freedom pairs $\left(\frac{11}{8},\frac{11}{8}\right)$ and $\left(0,\frac{7}{2}\right)$, respectively, can be obtained. The convex hull of the achieved degrees of freedom pairs gives the achievable degrees of freedom region which is shown in Fig~\ref{fig:nonidentical_2users_mac_1}.

Next, consider a two-transmitter multiple access channel with $\M_1=3, \M_2=2, \N=4, \T_1=8, \T_2=24$. In this case, the achievable and the outer degrees of freedom regions are shown in Fig.~\ref{fig:nonidentical_2users_mac_2}.

\section{Conclusion and Discussion}
\label{section:conclusion}

In this paper, multi-user networks without CSIT or CSIR are studied. For a broadcast channel where the receivers have identical coherence times, it was shown that the degrees of freedom region is tight against TDMA. However, when the receivers have heterogeneous coherence times, TDMA is no longer optimal since the difference of the coherence times can be a source of degrees of freedom gain, called coherence diversity. For a broadcast channel where the receivers coherence times are integer multiples of each other, achievable degrees of freedom gains were obtained using the product superposition scheme. Furthermore, an outer degrees of freedom region was obtained using channel enhancement where the receivers coherence times were increased so that the receivers of the enhanced channel have identical coherence times. As long as the coherence time is at least twice the number of transmit and receive antennas, the optimality of the achievable scheme was shown in four cases: when the transmitter has fewer antennas than the receivers, when all receivers have the same number of antennas, when the coherence times of the receivers are very long compared to the coherence time of one receiver, or the receivers have the same coherence times. For general coherence times that can be unaligned, product superposition transmission was extended, achieving achievable degrees of freedom region. Furthermore, a transmission scheme that combines product superposition and blind interference alignment was introduced in the staggered case. 

Multiple access channel with identical coherence times is studied, where a pilot-based achievable scheme was shown to be sum degrees of freedom optimal. Furthermore, a multiple access channel with heterogeneous coherence times is considered. When the transmitters coherence times are integer multiples of each other, an achievable pilot-based inner bound and an outer bound were obtained. The outer bound was obtained using channel enhancement where the transmitters coherence times were increased so that the transmitters of the enhanced channel have identical coherence times.

\appendices

\section{Coherent Broadcast Channel With Identical Coherent Times}
\label{appendix:identical_bc}

In the sequel, we show the degrees of freedom optimality of TDMA inner bound when the receivers have identical coherence times and CSI is assumed to be available at the receiver. We enhance the channel by providing global CSI at the receivers. Without loss of generality, assume that $\N_1 \leq \cdots \leq \N_\K$. When $ \M \leq \N_1$, the cooperative outer bound~\cite{Sato_outer} is tight against the TDMA inner bound. When $ \M > \N_1$, the broadcast channel is degraded~\cite{Bergmans_random}, hence,
\begin{align}
\R_i \leq & I \left( U_i ; \Y_i | \Hset, U^{i-1}  \right) \nonumber \\
       =  & I \left( \X ; \Y_i | \Hset, U^{i-1} \right) - I \left( \X ; \Y_i| \Hset, U^i \right),
\end{align}
where $U^i=\left\{ U_j \right\}_{j=1}^i$ is a set of auxiliary random variables such that $U_1 \rightarrow \cdots \rightarrow U_{\K-1} \rightarrow \X \rightarrow \left( \Y_1, \cdots \Y_\K \right)$ forms a Markov chain and for notational convenience we introduced a trivial random variable $U_0$ and $U_\K=\X$. $\Hset$ is the set of all channels. Furthermore,
\begin{align}
\R_1 \leq & \left( \N_1 - r_1 \right) \log \left( \rho \right) + \littleO(\log (\rho)), \nonumber \\
\R_i \leq & I \left( \X ; \Y_i | \Hset, U^{i-1} \right) - r_i \log \left( \rho \right) + \littleO(\log (\rho)), \quad i \neq 1,\K, \nonumber\\
\R_\K \leq & I \left( \X ; \Y_\K | \Hset, U^{\K-1} \right),
\end{align}
since the degrees of freedom of $I \left( \X ; \Y_1 | \Hset \right)$ is bounded by the single-receiver bound, i.e. $\N_1$, and $r_i$ is defined to be the degrees of freedom of the term $I \left( \X ; \Y_i| \Hset, U^i \right)$, where $ 0 \leq r_i \leq \N_i^{\ast}$. The extension of~\cite[Lemma~1]{Huang_degrees} to the $\K$-receiver case is straight forward, and hence, we can write 
\begin{equation}
I \left( \X ; \Y_{i,1} | \Hset, U^i, \Y_{i,2:\N_i^{\ast}} \right) \leq \frac{r_i}{\N_i^{\ast}} \log \left( \rho \right) + \littleO(\log (\rho)),
\label{eq:DoFPerAntennaCoherent}
\end{equation}
where $\Y_{i,1} \in \mathbb{C}^{1 \times \T}$ is the signal at antenna~1 of receiver~$i$ over the entire $\T$-length coherence time whereas $\Y_{i,2:\N_i^{\ast}} \in \mathbb{C}^{(\N_i^{\ast}-1) \times \T}$ is the matrix consists of the signal at antennas~$2,3,\ldots, \N_i^{\ast}$ of receiver~$i$ over the entire $\T$-length coherence time. Furthermore,
\begin{align}
I \left( \X ; \Y_i | \Hset, U^{i-1} \right) 
\overset{(a)}{=} & I\left(\X ;\Y_{i,1:\N_i^{\ast}}|\Hset,U^{i-1}\right) + I \left(\X ;\Y_{i,\N_i^{\ast}+1:\N_i}| \Hset, U^{i-1}, \Y_{i,1:\N_i^{\ast}} \right) \nonumber\\
\overset{(b)}{=} & I\left(\X ;\Y_{i,1:\N_{i-1}^{\ast}}| \Hset, U^{i-1} \right) \twocolbreak + I \left( \X ; \Y_{i,\N_{i-1}^{\ast}+1:\N_i^{\ast}}|\Hset,U^{i-1}, \Y_{i,1:\N_{i-1}^{\ast}} \right) \nonumber \\ & + \littleO(\log (\rho)) \nonumber \\
\overset{(c)}{=} & I\left(\X ;\Y_{i-1,1:\N_{i-1}^{\ast}}|\Hset,U^{i-1} \right) \twocolbreak
+ I\left(\X ;\Y_{i,\N_{i-1}^{\ast}+1:\N_i^{\ast}}| \Hset, U^{i-1},\Y_{i-1,1:\N_{i-1}^{\ast}} \right)  \nonumber \\ & + \littleO(\log (\rho)) \nonumber \\
= & r_{i-1} \log \left(\rho\right) \twocolbreak + \!\!\!\!\!\!\!\! \sum_{j=\N_{i-1}^{\ast}+1}^{\N_i^{\ast}} I\left(\X ;\Y_{i,j} |\Hset, U^{i-1},\Y_{i-1,1:\N_{i-1}^{\ast}},\Y_{i,j+1:\N_i^{\ast}} \right) +\littleO(\log (\rho)) \nonumber\\
\overset{(d)}{\leq} & r_{i-1}\log \left(\rho \right) \twocolbreak + \left(\N_i^{\ast}-\N_{i-1}^{\ast}\right) I \left(\X ;\Y_{i-1,1}|\Hset,U^{i-1},\Y_{i-1,2:\N_{i-1}^{\ast}}\right) + \littleO(\log (\rho)) \nonumber \\
\overset{(e)}{\leq} & r_{i-1} \log\left(\rho \right) +\left(\N_i^{\ast}-\N_{i-1}^{\ast}\right) \frac{r_{i-1}}{\N_{i-1}^{\ast}} \log \left( \rho \right) + \littleO(\log (\rho)) \nonumber \\
\leq & \frac{\N_i^{\ast}}{\N_{i-1}^{\ast}} r_{i-1} \log \left( \rho \right) + \littleO(\log (\rho)),
\end{align}
where $(a)$ and $(b)$ follow from applying the chain rule, and $ h \left(\Y_{i,\N_i^{\ast}+1:\N_i} | \Hset, U^{i-1}, \Y_{i,\N_i^{\ast}} \right) = \littleO(\log (\rho))$. Furthermore, $(c)$ follows since $\Y_{i,1:\N_{i-1}^{\ast}}$ and $\Y_{i-1,1:\N_{i-1}^{\ast}}$ are statistically the same. $(d)$ follows from applying the straight forward extension of~\cite[Lemma 1]{Huang_degrees} and $(e)$ follows from~\eqref{eq:DoFPerAntennaCoherent}. Therefore,
\begin{align}
\d_1 \leq & \N_1 - r_1, \nonumber \\
\d_i \leq & \frac{\N_i^{\ast}}{\N_{i-1}^{\ast}} r_{i-1} - r_i, \quad i\neq 1,\K, \nonumber \\
\d_\K \leq & \frac{\N_\K^{\ast}}{\N_{\K-1}^{\ast}} r_{\K-1}, 
\end{align}
which gives the region defined in~\eqref{eq:identical_bc}.

\section{Proof of Lemma~\ref{lemma:enhance}}
\label{appendix:enhance}

Consider the set of receivers $\Jset \subseteq \left[ 1:\K \right]$ where the receivers are ordered ascendingly according to the coherence times length, i.e., $\T_j \geq \T_{j-1}, \forall j \in \Jset$. The proof consists of two steps. First, we show that the individual degrees of freedom of each receiver is non-decreasing with the increase of the coherence time of this receiver. Second, we show that the degrees of freedom region of the channel is non-decreasing with the increase of the coherence times of the receivers. For the first step of the proof we introduce the following Lemma.

\begin{lemma}\label{lemma:nondec_T}
For the broadcast channel considered in Section~\ref{section:nonidentical_bc}, define $\Jset= \left\{ i_1, \cdots, i_\J \right\} \subseteq \left[ 1:\K \right]$ with $\T_j \geq \T_{j-1}, \forall j \in \Jset$ and $\Psi_j$ as the message of receiver~$j \in \Jset$. Thus, we have
\begin{equation}\label{eq:nondec_T}
\N_j^{\ast} \left( \frac{1}{\J} - \frac{\N_j^{\ast}}{\T_j}\right) \leq 
\text{MG} \left\{ \frac{1}{\bar{n}} I\left( \Psi_j ; \Y_j^n \right) \right\} 
\leq \N_j^{\ast} \left( 1 - \frac{\N_j^{\ast}}{\T_j}\right), 
\end{equation}
where $\text{MG}(x)$ is the multiplexing gain of a function $x(\rho)$ of $\rho$ and defined as
\begin{equation}
\text{MG}(x) = \lim_{\rho \rightarrow \infty} \sup \frac{x(\rho)}{\log (\rho)}.
\end{equation}
\end{lemma}

\begin{proof}
We first prove the right inequality of~\eqref{eq:nondec_T}. We have 
\begin{align}
\text{MG} \left\{ \frac{1}{\bar{n}} I\left( \Psi_j ; \Y_j^n \right) \right\} & \overset{(a)}{\leq} \text{MG} \left\{ \frac{1}{\bar{n}} I\left( \X^n ; \Y_j^n \right) \right\} \nonumber \\
                            & \overset{(b)}{\leq} \N_j^{\ast} \left( 1 - \frac{\N_j^{\ast}}{\T_j}\right),
\end{align}
where $(a)$ follows from the data processing inequality and $(b)$ follows from the single-receiver results~\cite{Zheng_communication}. Next, we show the left inequality of~\eqref{eq:nondec_T}. Assume that we have the following transmitted sequence 
\begin{equation}
\overline{\X}^n = \left[ \U_{i_1}, \cdots, \U_{i_\J} \right],
\end{equation}
where $\U_j \in \mathbb{C}^{\N_j \times \frac{\bar{n}}{\J}}$ is the matrix containing the signal of receiver~$j \in \Jset$ and the matrix is constructed to be on the form of the optimal input of a non-coherent single receiver~\cite{Zheng_communication}. Hence,
\begin{align}
\text{MG} \left\{ \frac{1}{\bar{n}} I\left( \Psi_j ; \Y_j^n \right) \right\} & \geq \text{MG} \left\{ \frac{1}{\bar{n}} I\left( U_j ; \Y_j^n \right) \right\} \twocolbreak
\geq \frac{\N_j^{\ast}}{\J} \left( 1 - \frac{\N_j^{\ast}}{\T_j}\right) \twocolbreak
\geq       \N_j^{\ast} \left( \frac{1}{\J} -  \frac{\N_j^{\ast}}{\T_j}\right).
\end{align}
Thus, the proof of Lemma~\ref{lemma:nondec_T} is completed.
\end{proof}

By Lemma~\ref{lemma:nondec_T}, we have lower and outer bounds which are increasing with $\T_j$. Furthermore, the gap between the two bounds is 
\begin{equation}
\Delta = \N_j^{\ast} \left( 1 - \frac{\N_j^{\ast}}{\T_j}\right) 
- \N_j^{\ast} \left( \frac{1}{\J} - \frac{\N_j^{\ast}}{\T_j}\right) = \N_j^{\ast} \frac{\J - 1}{\J}.
\end{equation}
Therefore, $\text{MG} \left\{ \frac{1}{\bar{n}} I\left( \Psi_j ; \Y_j^n \right) \right\}$ is non-decreasing with the increase of $\T_j$, which completes the first step of the proof of Lemma~\ref{lemma:enhance}. 

Now, we give the second part of the proof via a contradiction argument. Define $\mathcal{D}$ to be the degrees of freedom region of a set of receivers with unequal coherence times where $\max_j \T_j = \T_{\max}$, and $\Y_j^n$ denotes the signal at receiver~$j$. Define $\mathcal{\overline{D}}$ to be the degrees of freedom region of the receivers where the coherence times of all receivers is $\T_{\max}$,  where $\bar{\Y}_j^n$ denotes the signal at receiver~$j$ of this enhanced channel. Define $\tilde{D} \in \mathcal{D}$ to be a degrees of freedom tuple, and $d_j \in \tilde{D}$ is the degrees of freedom of receiver~$j$. Assume that $\tilde{D} \notin \mathcal{\overline{D}}$. By Fano's inequality, 
\begin{align}
d_j = & \text{MG} \left\{ \frac{1}{\bar{n}} I \left(\Psi_j ; \Y_j^n \right) \right\}, \nonumber \\
    \leq & \text{MG} \left\{ \frac{1}{\bar{n}} I \left(\Psi_j ; \bar{\Y}_j^n \right) \right\},
\end{align}
where $\Psi_j$ is the message of receiver~$j \in \Jset$, and the last inequality follows from Lemma~\ref{lemma:nondec_T}. $\text{MG} \left\{ \frac{1}{\bar{n}} I \left(\Psi_j ; \bar{\Y}_j^n \right) \right\} \in \mathcal{\overline{D}}$, therefore, $d_j \in \mathcal{\overline{D}}, \forall j$, which contradicts the initial assumption completing the second part of the proof.

\section{Proof of Lemma~\ref{lemma:enhance_mac}}
\label{appendix:enhance_mac}

Consider the set of transmitters $\Jset= \left\{ i_1, \cdots, i_\J \right\} \subseteq \left[ 1:\K \right]$ where $\forall j \in \Jset, \frac{\T_j}{\T_{j-1}}\in {\mathbb Z}$. By Fano's inequality, as $\bar{n} \rightarrow \infty$,
\begin{equation}
\sum_{j \in \Jset} R_j \leq \frac{1}{\bar{n}} I \left(\X_{i_1}^n, \cdots, \X_{i_J}^n ; \Y^n \right),
\end{equation}
where $\X_j^n$ is transmitter~$j \in \Jset$ signal and $\Y^n$ is the received signal over the entire transmission time of length $\bar{n}$ slots. In the sequel, we show that the degrees of freedom of $\frac{1}{\bar{n}} I \left(\X_j ; \Y^n \right)$ is non-decreasing in $\T_j$. We give lower and upper bounds on this term, and furthermore, both bounds are non-decreasing in $\T_j$.

\begin{lemma}\label{lemma:nondec_T_mac}
Consider the multiple access channel in Section~\ref{section:nonidentical_mac}. Define $\Jset= \left\{ i_1, \cdots, i_\J \right\} \subseteq \left[ 1:\K \right]$ and $\Psi_j$ as the message of transmitter~$j \in \Jset$. Thus, we have 
\begin{align}\label{eq:nondec_T_mac}
\sum_{j \in \Jset} \M_j^{\ast} \left( \frac{1}{\J} - \frac{\M_j^{\ast}}{\T_j}\right) 
& \leq \text{MG} \left\{ \frac{1}{\bar{n}} I\left( \X_{i_1}^n, \cdots, \X_{i_J}^n ; \Y^n \right) \right\} \twocolbreak 
  \leq \sum_{j \in \Jset} \M_j^{\ast} \left( 1 - \frac{\M_j^{\ast}}{\T_j}\right), 
\end{align}
where $\M_j^{\ast} = \min\left\{ \M_j, \N\right\}$.
\end{lemma}

\begin{proof}
We first prove the right inequality of~\eqref{eq:nondec_T_mac}. We have 
\begin{align}
\text{MG} \left\{ \frac{1}{\bar{n}} I\left( \X_{i_1}^n, \cdots, \X_{i_J}^n ; \Y^n \right) \right\} 
            & \leq \sum_{j \in \Jset} \M_j^{\ast} \left( 1 - \frac{\M_j^{\ast}}{\T_j}\right),
\end{align}
where the above inequality follows from the single-transmitter results~\cite{Zheng_communication}. Next, we show the left inequality of~\eqref{eq:nondec_T_mac}. Assume that we have the following transmitted sequence 
\begin{equation}
\overline{\X}_j^n = \left[ 0, \cdots, 0, \U_j, 0, \cdots, 0 \right],
\end{equation}
where $\U_j \in \mathbb{C}^{\M_j \times \frac{\bar{n}}{\J}}$ is the matrix containing the signal of transmitter~$j \in \Jset$ and the matrix is constructed to be on the form of the optimal input of a non-coherent single transmitter~\cite{Zheng_communication}. Hence,
\begin{align}
\text{MG} \left\{ \frac{1}{\bar{n}} I\left( \X_{i_1}^n, \cdots, \X_{i_J}^n ; \Y^n \right) \right\} 
& \geq \text{MG} \left\{ \frac{1}{\bar{n}} \sum_{j \in \Jset} I\left( U_j ; \Y_j^n \right) \right\} \nonumber \\ 
& \geq \sum_{j \in \Jset} \frac{\M_j^{\ast}}{\J} \left( 1 - \frac{\M_j^{\ast}}{\T_j}\right) \nonumber \\
& \geq       \sum_{j \in \Jset} \M_j^{\ast} \left( \frac{1}{\J} -  \frac{\M_j^{\ast}}{\T_j}\right).
\end{align}
Thus, the proof of Lemma~\ref{lemma:nondec_T_mac} is completed.
\end{proof}

By Lemma~\ref{lemma:nondec_T_mac}, we have lower and outer bounds which are increasing with $\T_j$. Furthermore, the difference between the two bounds is 
\begin{equation}
\Delta = \sum_{j \in \Jset} \M_j^{\ast} \left( 1 - \frac{\M_j^{\ast}}{\T_j}\right) 
- \M_j^{\ast} \left( \frac{1}{\J} - \frac{\M_j^{\ast}}{\T_j}\right) = \M_j^{\ast} \frac{\J - 1}{\J}.
\end{equation}
Therefore, $\text{MG} \left\{ \frac{1}{\bar{n}} I\left( \X_{i_1}^n, \cdots, \X_{i_J}^n ; \Y^n \right) \right\}$ is non-decreasing with the increase of $\T_j$, and hence, the proof of Lemma~\ref{lemma:enhance_mac} is completed.

\bibliographystyle{IEEEtran}
\bibliography{IEEEabrv,J_IT2015}

\end{document}